\documentclass[a4paper,10pt]{article}

\usepackage{graphicx}
\usepackage{color}

\usepackage{mathrsfs}
\usepackage{bussproofs}

\usepackage{subfigure}

\newcommand{\condnocaptionrule}{}

\usepackage{amsfonts}

\usepackage{url}
\usepackage{proof}       
\usepackage{amsmath}
\usepackage{amssymb}
\usepackage{cmll}
\usepackage{stmaryrd}
\newcommand{\ket}[1]{|#1\rangle}
  
\urldef\mailugo\url{dallago@cs.unibo.it}
\urldef\mailclaudia\url{faggian@pps.jussieu.fr}
\urldef\mailichiro\url{ichiro@is.s.u-tokyo.ac.jp}
\urldef\mailakira\url{yoshimizu@is.s.u-tokyo.ac.jp}

\usepackage{a4wide}

\newenvironment{varitemize}
{
\begin{list}{\labelitemi}
{\setlength{\itemsep}{0pt}
 \setlength{\topsep}{0pt}
 \setlength{\parsep}{0pt}
 \setlength{\partopsep}{0pt}
 \setlength{\leftmargin}{15pt}
 \setlength{\rightmargin}{0pt}
 \setlength{\itemindent}{0pt}
 \setlength{\labelsep}{5pt}
 \setlength{\labelwidth}{10pt}
}}
{
 \end{list} 
}

\newcounter{number}

\newenvironment{varenumerate}
{\begin{list}{\arabic{number}.}
  {
   \usecounter{number}
   \setlength{\labelwidth}{4.0mm}
   \setlength{\labelsep}{2.0mm}
   \setlength{\itemindent}{0.0mm}
   \setlength{\itemsep}{0.0mm}
   \setlength{\topsep}{0.0mm}
   \setlength{\parskip}{0.0mm}
   \setlength{\parsep}{0.0mm}
   \setlength{\partopsep}{0.0mm}
  }
}
{\end{list}}

\newcommand{\MLL}{\ensuremath{\textsf{MLL}}}

\newcommand{\NN}{\mathbb{N}}

\newcommand{\typeone}{A}
\newcommand{\typetwo}{B}

\newcommand{\typefour}{D}

\newcommand{\atomone}{\alpha}

\newcommand{\tens}{\otimes}
\newcommand{\lin}{\multimap}

\newcommand{\ctxthree}{\Theta}

\newcommand{\cut}{\mathsf{C}}
\newtheorem{theorem}{Theorem}[section]   % Numbered within each section
\newtheorem{lemma}[theorem]{Lemma}
\newtheorem{prop}[theorem]{Proposition}

\newtheorem{deff}[theorem]{Definition}
\newtheorem{rem}[theorem]{Remark}

\newtheorem{corollary}[theorem]{Corollary}
\newenvironment{proof}{\begin{trivlist}
       \item[\hskip \labelsep {\bfseries Proof.}]}{\hfill $\Box$ \end{trivlist}}
\newcommand{\termone}{t}
\newcommand{\termtwo}{u}
\newcommand{\termthree}{v}

\newcommand{\varone}{x}
\newcommand{\vartwo}{y}
\newcommand{\varthree}{z}
\newcommand{\varfour}{w}
\newcommand{\qvarone}{r}
\newcommand{\qvartwo}{q}
\newcommand{\new}{\mathsf{new}}
\newcommand{\tfalse}{\mathsf{ff}}
\newcommand{\ttrue}{\mathsf{tt}}
\newcommand{\meas}{\mathsf{meas}}

\newcommand{\bool}{\mathbb{B}}
\newcommand{\qbool}{\mathbb{Q}}
\newcommand{\typone}{A}
\newcommand{\typtwo}{B}
\newcommand{\typthree}{C}

\newcommand{\contone}{\Gamma}
\newcommand{\conttwo}{\Delta}

\newcommand{\tyg}[3]{#1\vdash #2:#3}

\newcommand{\subst}[3]{#1\{#3/#2\}}

\renewcommand{\b}{{}^{\bot}}

\newcommand{\lam}{\lambda}
\newcommand{\abstr}[2]{\lambda #1.#2}

\newcommand{\one}{1}

\newcommand{\imp}{\multimap}
\renewcommand{\lin}{\multimap}
\newcommand{\pair}[2]{\langle #1,#2 \rangle}

\newcommand{\llet}[4]{\mathsf{let}\;\langle #1,#2\rangle\;\mathsf{be}\;#3\;\mathsf{in}\;#4}

\newcommand{\ifthen}[3]{\mathsf{if}\;#1\;\mathsf{then}\;#2\;\mathsf{else}\;#3}

\newcommand{\ecimp}{\Box\!\!\Rightarrow} 

\usepackage{tikz}
\usetikzlibrary{calc}
\usetikzlibrary{backgrounds}
\usetikzlibrary{decorations.shapes}
\usetikzlibrary{decorations.text}
\usetikzlibrary{decorations.pathmorphing}
\usetikzlibrary{decorations.markings}
\usetikzlibrary{matrix}
\usetikzlibrary{decorations.pathreplacing}
\usetikzlibrary{arrows}
\usetikzlibrary{positioning}
\usepackage[nofancy]{tikz-inet}

\newcommand{\ax}{\mathsf{ax}}
\tikzset{%
  ocenter/.style={baseline={([yshift=-.5ex, xshift=-.5ex]current bounding box)}},  
  nospace/.style={inner sep= 0pt},
  etic/.style={inner sep= 0.5pt, fill=white, anchor= center},
  nopol/.style={->, shorten <=0.5pt, shorten >=0.5pt, draw=gray, line width=0.18ex},
  nopolrev/.style={<-, shorten <=1pt, shorten >=1pt, draw=gray, line width=0.18ex},
  nopolgen/.style={preaction={decorate},decoration={markings,mark=at position .5 with {\draw [shorten >=0pt, shorten <=0pt,draw=gray,-](0pt,3pt) -- (0pt,-3pt);}},->, shorten <=0.5pt, shorten >=0.5pt, draw=gray, line width=0.18ex},
  nopolrevgen/.style={preaction={decorate},decoration={markings,mark=at position .5 with {\draw [shorten >=0pt, shorten <=0pt,draw=gray,-](0pt,3pt) -- (0pt,-3pt);}},<-, shorten <=1pt, shorten >=1pt, draw=gray, line width=0.18ex},
  pos/.style={->, shorten >=0.5pt, shorten <=0.5pt, draw=blue, line width=0.18ex},
  neg/.style={->, densely dotted, shorten >=0.5pt, shorten <=0.5pt, draw=red, line width=0.18ex, overlay},
  jboxline/.style={draw= gray,rounded corners, line width=0.20ex, overlay},
  exboxline/.style={draw= gray,line width=0.20ex, overlay},
  noboxline/.style={draw= white,rounded corners, line width=0ex},
  net/.style={draw=gray,inner sep=2pt,thick,ellipse, anchor=center, font=\scriptsize},
  reg/.style={draw=gray,inner sep=2pt,ultra thick,rectangle, anchor=center, font=\scriptsize},
  every label/.style={label distance = 1pt, font=\tiny, inner sep= 1pt},  
  every node/.style={font=\scriptsize }%\small
}
\newcommand{\altax}{8pt}

\newcommand{\stalt}{22pt}
\newcommand{\stlar}{15pt}
\newcommand{\hstalt}{\stalt/2}
\newcommand{\hstlar}{\stlar/2}

\newcommand{\ilar}{12pt}

\newcommand{\lzeroary}[5]{
\node at (#1.center) [above= #4, etic] (#2){#3};
\draw[#5] (#2) to (#1);}

\newcommand{\lunarysymbol}[4]{
\node at ($(#2.center) ! .5 ! (#1.center)$) [etic] (#3){ #4};
}

\newcommand{\lunary}[5]{
\lunarysymbol{#1}{#2}{#3}{#4}

\draw[#5] (#2) to (#3);
\draw[#5] (#3) to (#1);
}

\newcommand{\lbinsym}[5]{
\node at ($(#2.center) ! .5 ! (#3.center) ! .5 ! (#1.center)$) [etic] (#4){#5};
}

\newcommand{\lbinedgesabove}[4]{
\draw[#4, in=150, out=-90] (#1) to (#3);
\draw[#4, in=30, out=-90] (#2) to (#3);
}

\newcommand{\lbinary}[6]{
\lbinsym{#1}{#2}{#3}{#4}{#5}

\lbinedgesabove{#2}{#3}{#4}{#6}

\draw[#6] (#4) to (#1);
}

\newcommand{\stboxlw}{6pt}
\newcommand{\stboxh}{20pt}

\newcommand{\stboxrwaux}{35pt}

\newcommand{\boxnodes}[4]{
\node at (#1.center)[left = #2,nospace](#1so){};
\node at (#1.center)[right = #3,nospace](#1se){};
\node at (#1se.center)[above=#4,nospace](#1ne){};
\node at (#1ne-|#1so)[nospace](#1no){};}

\newcommand{\boxline}[2]{
\draw[#2](#1.center) -- (#1se.center) -- (#1ne.center) -- (#1no.center) -- (#1so.center)--(#1.center);}

\newcommand{\abox}[5]{
\boxnodes{#1}{#3}{#4}{#5}
\boxline{#1}{#2line}}

\newcommand{\ltens}[4]{
\lbinary{#1}{#2}{#3}{#4}{$\tens$}{nopol}}

\newcommand{\lpar}[4]{
\lbinary{#1}{#2}{#3}{#4}{$\parr$}{nopol}}

\newcommand{\lone}[2]{
\lzeroary{#1}{#2}{\tiny $\mathsf{one}$}{\hstalt}{nopol}}

\newcommand{\lbot}[2]{
\lzeroary{#1}{#2}{\tiny $\mathsf{bot}$}{\hstalt}{nopol}}

\newcommand{\lsync}[3]{
\lunary{#1}{#2}{#3}{\tiny $\blacksquare$}{nopol}
}

\newlength{\figurewidth}
\begin{document}

\renewcommand{\cut}{\mathsf{cut}}

\newcommand{\IAM}{\textsf{IAM}}
\newcommand{\QSIAM}{\textsf{QSIAM}}
\newcommand{\SIAM}{\textsf{SIAM}}
\newcommand{\varalphabet}{\mathcal A}

\newcommand{\SMLL}{\textsf{SMLL}}
\newcommand{\SMLLb}{\textsf{SMLL}}
\newcommand{\SMLLcf}{\textsf{SMLL}${}^0$}

\newcommand{\QSMLL}{\textsf{QSMLL}}

\newcommand{\botlk}{\mathsf{bot}}
\newcommand{\onelk}{\mathsf{one}}
\newcommand{\clk}{\mathsf{c}}
\newcommand{\axlk}{\mathsf{ax}}
\newcommand{\boxlk}{\mathsf{box}}

\newcommand{\up}{\uparrow}
\newcommand{\down}{\downarrow}
\newcommand{\at}{\gamma}

\newcommand{\before}{\less_p}

\newcommand{\formone}{A}
\newcommand{\formtwo}{B}
\newcommand{\formthree}{C}

\newcommand{\bnf}{::=}
\newcommand{\midd}{\; \; \mbox{\Large{$\mid$}}\;\;}

\newcommand{\atone}{\gamma}
\newcommand{\attwo}{\delta}

\newcommand{\posfone}{P}
\newcommand{\negfone}{N}

\newcommand{\netone}{R}
\newcommand{\nettwo}{R'}

\newcommand{\swone}{s}

\newcommand{\rednet}{\longrightarrow}

% % % % % % % % % % % % % % % % % % % % % % %
%\newcommand{\ONE}{{\tt ONES}}
\newcommand{\ONE}{\mathtt{ONES}}
\newcommand{\BOTBOX}{\mathtt{BOTBOX}}

\newcommand{\Pos}{\mathtt{POS}}

\newcommand{\PosI}{\mathtt{INIT}}
\newcommand{\PosF}{\mathtt{FIN}}

\newcommand{\pos}[1]{\mathtt{POS}(#1)}
\newcommand{\botcon}[1]{\mathtt{INIT}(#1)}  
%per me, le posizioni iniziali e finali possono anche essere atomi
\newcommand{\onecon}[1]{\mathtt{FIN}(#1)}
\newcommand{\botbox}[1]{\mathtt{BOTBOX}(#1)}
\newcommand{\oneany}[1]{\mathtt{ONES}(#1)}
\newcommand{\tksone}{\mathbf{T}}
\newcommand{\pssetone}{L}   %il problema di S, e' che e' variabile per proof-nets
\newcommand{\pssettwo}{Y}
\newcommand{\posone}{\mathbf{o}}
\newcommand{\postwo}{\mathbf{p}}
\newcommand{\posthree}{\mathbf{q}}
\newcommand{\states}[1]{\mathit{ST}(#1)}
\newcommand{\istates}[1]{\mathit{IST}(#1)}
\newcommand{\fstates}[1]{\mathit{FST}(#1)}
\newcommand{\trrel}{\hookrightarrow}
\newcommand{\trrelmult}{\looparrowright}
\newcommand{\stone}{s}
\newcommand{\sttwo}{t}
\newcommand{\sem}[1]{[\![#1]\!]}

% % % % % % % % % % % % % % % % % % %

\newcommand{\condmedskip}{}
%\hyphenation{op-tical net-works semi-conduc-tor}

%%%%%%%%%%%%%%%%%%%%%%%%%%%%%%%%%%%%%%%%%%%%%%%%%%%%%%
\title{The Geometry of Synchronization (Long Version)}
%%%%%%%%%%%%%%%%%%%%%%%%%%%%%%%%%%%%%%%%%%%%%%%%%%%%%%
\author{
  Ugo Dal Lago\footnote{University of Bologna \& INRIA, \mailugo}
  \and
  Claudia Faggian\footnote{CNRS, \mailclaudia}
  \and
  Ichiro Hasuo\footnote{University of Tokyo, \mailichiro}
  \and
  Akira Yoshimizu\footnote{University of Tokyo, \mailakira}}
\date{}

\maketitle

\begin{abstract}
We graft synchronization onto Girard's Geometry of Interaction in its
most concrete form, namely token machines. This is realized by
introducing proof-nets for \SMLL, an extension of multiplicative
linear logic with a specific construct modeling synchronization
points, and of a multi-token abstract machine model for it.
Interestingly, the correctness criterion ensures the absence of
deadlocks along reduction and in the underlying machine, this way
linking logical and operational properties.
\end{abstract}

%%%%%%%%%%%%%%%%%%%%%%%%
\section{Introduction}
%%%%%%%%%%%%%%%%%%%%%%%%
One of the reasons making Linear Logic~\cite{Girard87} a breakthrough
not only in proof theory but also in programming language semantics,
is that it enables an interactive view of computation through Game
Semantics and the Geometry of Interaction (GoI in the following). This
way, proofs and higher-order programs are seen as mathematical objects
with a rich interactive behaviour (composition in the former,
execution in the latter). One could say that Game Semantics focuses on
the interpretation of programs, the tool of election being full
abstraction, while Geometry of Interaction is a fine-grained model of
computation itself, and provides insights on quantitative aspects of
computation~\cite{BaillotPedicini,DalLago}, tools to allow
optimization~\cite{GAL} and guidelines in the design of a
compiler~\cite{Mackie,Pinto}. In either settings, one can describe the
dynamics of the interaction between programs and their environments in
several ways. In Game Semantics, this can take a categorical form, or
the operational form of an abstract machine, as in
\cite{CurienHerbelin}.  Similarly in GoI, the interaction can be
described either via automata~\cite{DanosRegnier}, or via traced
monoidal categories~\cite{Abramsky}. GoI can also be presented
algebraically, by way of operator algebras~\cite{GirardV}, or more
operationally as an algebra of clauses~\cite{GirardIII}.  Different
presentations suit different aims.

In this paper, we are interested in the most concrete presentation of
GoI, and in particular in the so-called Interaction Abstract Machines
(IAMs in the following). IAMs are bi-deterministic automata by which
one interprets $\lambda$-terms in such a way that $\beta$-equivalent
terms are interpreted by equivalent IAMs (i.e., IAMs computing the
same function).  A single run of the IAM interpreting a $\lambda$-term
does \emph{not} suffice to capture completely the behaviour of the
term itself: this in general requires multiple calls to the IAM, which
however can proceed in parallel, without any need for
synchronization. In a sense, this shows that GoI has the potential to
somehow capture the inherent parallelism of functional programs, and
this has been indeed exploited as a compilation technique through
(directed) virtual reduction~\cite{Pinto,Pedicini} (even though such
work deviates from purely interactive machines). The captured kind of
parallelism is however lacking a fundamental ingredient, since the
parallel components are not allowed to interact in non-trivial ways.

This work is a study of synchronization in the context of linear logic
proofs. More specifically, the contributions of this paper are
threefold:
\begin{varitemize}
\item 
  Proof-nets of multiplicative linear logic (\MLL\ in the
  following) are enriched so as to include a specific rule
  representing synchronization points. This is done by extending the
  kinds of links on top of which proof-structures are defined, then
  properly adapting Danos and Regnier's correctness criterion. The
  resulting system, called \SMLL, is shown to enjoy cut elimination.
\item 
  A specific kind of Interaction Abstract Machine, called \SIAM,
  is introduced and shown to be a model of \SMLL. Remarkably, \SMLL\
  nets have the property that the underlying \SIAM\ is deadlock-free.
\item 
  \SMLL\ is shown to be sufficiently rich to interpret a quantum
  $\lambda$-calculus akin to those recently introduced by Selinger and
  Valiron~\cite{SelingerValiron}. Synchronization plays the essential
  role of reflecting quantum entanglement, itself a crucial ingredient
  for the efficiency of quantum computation~\cite{Jozsa97}. This
  requires to extend \SMLL\ only slightly, by endowing
  proof-structures with quantum registers, but keeping the underlying
  logical structure essentially unchanged.
\end{varitemize}
%%%%%%%%%%%%%%%%%%%%%%%%%%%%%%%%%%%%%%%%%%%%%%%%%%%%%%%%%%%%%%%%%%%%%%%
\section{Linear Logic and Token Machines}
%%%%%%%%%%%%%%%%%%%%%%%%%%%%%%%%%%%%%%%%%%%%%%%%%%%%%%%%%%%%%%%%%%%%%%%
\newcommand{\true}{\mathtt{true}}
\newcommand{\false}{\mathtt{false}}
\newcommand{\notc}{\mathtt{not}}
\newcommand{\booltype}[1]{\mathit{B}_{#1}}
\newcommand{\PCF}{\textsf{PCF}}
\newcommand{\suc}{\mathit{S}}
\newcommand{\zero}{\underline{\mathit{0}}}
\newcommand{\cnot}{\mathsf{CNOT}}
\newcommand{\hada}{\mathsf{H}}
In this section, we will give some hints about how IAMs (close
variations of which include
token machines~\cite{Mackie} and context semantics~\cite{DanosRegnier}) are defined, pointing to the relevant literature on the subject.

Let's start with linear, simply-typed, $\lambda$-calculus. Even in the absence of constants, the language has a decent expressive
power~\cite{Terui}: all boolean circuits can be encoded into it. Booleans can be encoded as the two permutations on a two-element
set: $\true$ is the $\lambda$-term $\abstr{\pair{\varone}{\vartwo}}{\pair{\varone}{\vartwo}}$, while $\false$ is the 
$\lambda$-term $\abstr{\pair{\varone}{\vartwo}}{\pair{\vartwo}{\varone}}$; both can be given the same type 
$\booltype{\atomone}=\atomone\tens\atomone\lin\atomone\tens\atomone=(\atomone\b\parr\atomone\b)\parr(\atomone\otimes\atomone)$. Boolean functions can also be represented
in the calculus. As a simple example, consider the combinator $\notc=\abstr{\varone}{\abstr{\pair{\vartwo}{\varthree}}
{\varone\pair{\varthree}{\vartwo}}}$. As can be easily verified, the application $\notc\;\true$ has type $\booltype{\atomone}$
and $\beta$-reduces to $\false$. There's a different, ``reduction-free'' way to compute the result of the application $\notc\;\true$:
traveling inside (a graph-based representation of) the term. Indeed, a type derivation for
$\false$ in Figure~\ref{fig:pfnet_false} can be represented as the proof-net in Figure~\ref{fig:examples_false}. If we start from the \emph{leftmost} 
negative occurrence of $\atomone$ in its conclusion and track it in the natural way, we end up in the \emph{rightmost} 
positive occurrence of $\atomone$ in the conclusion (Figure~\ref{fig:pfnet_false_traceOne}). Similarly, if we start 
from the rightmost negative occurrence of $\atomone$, we arrive at the
rightmost positive occurrence of $\atomone$ (Figure~\ref{fig:pfnet_false_traceTwo}). 
As expected, the term $\true$ behaves the opposite. Generalizing a bit, this game reveals the shape of normal forms, and
the nice thing is that it can be played on terms which \emph{are not} in normal form, this way becoming a fully-fledged notion of 
computation. As an example, consider the proof-net in Figure~\ref{fig:examples_nottrue} which corresponds to the term $\notc\;\true$, 
and the paths (like in Figure~\ref{fig:pfnet_false_traceOne}) certifying that the term rewrites to $\false$. The way one traces atom occurrences along paths can be formalized as 
an automaton, the Interaction Abstract Machine (IAM in the following), whose states are atom occurrences which are associated with the edges of the graph,  and whose transitions only depend on the nodes in the underlying graph (the transitions for \MLL\ can be found in the first two rows of Figure~\ref{SIAM}).
\begin{figure*}
\begin{center}
\fbox{
\begin{minipage}{.97\textwidth}
\begin{center}
\raisebox{11pt}{\subfigure[]{\raisebox{7pt}{\includegraphics[scale=0.13]{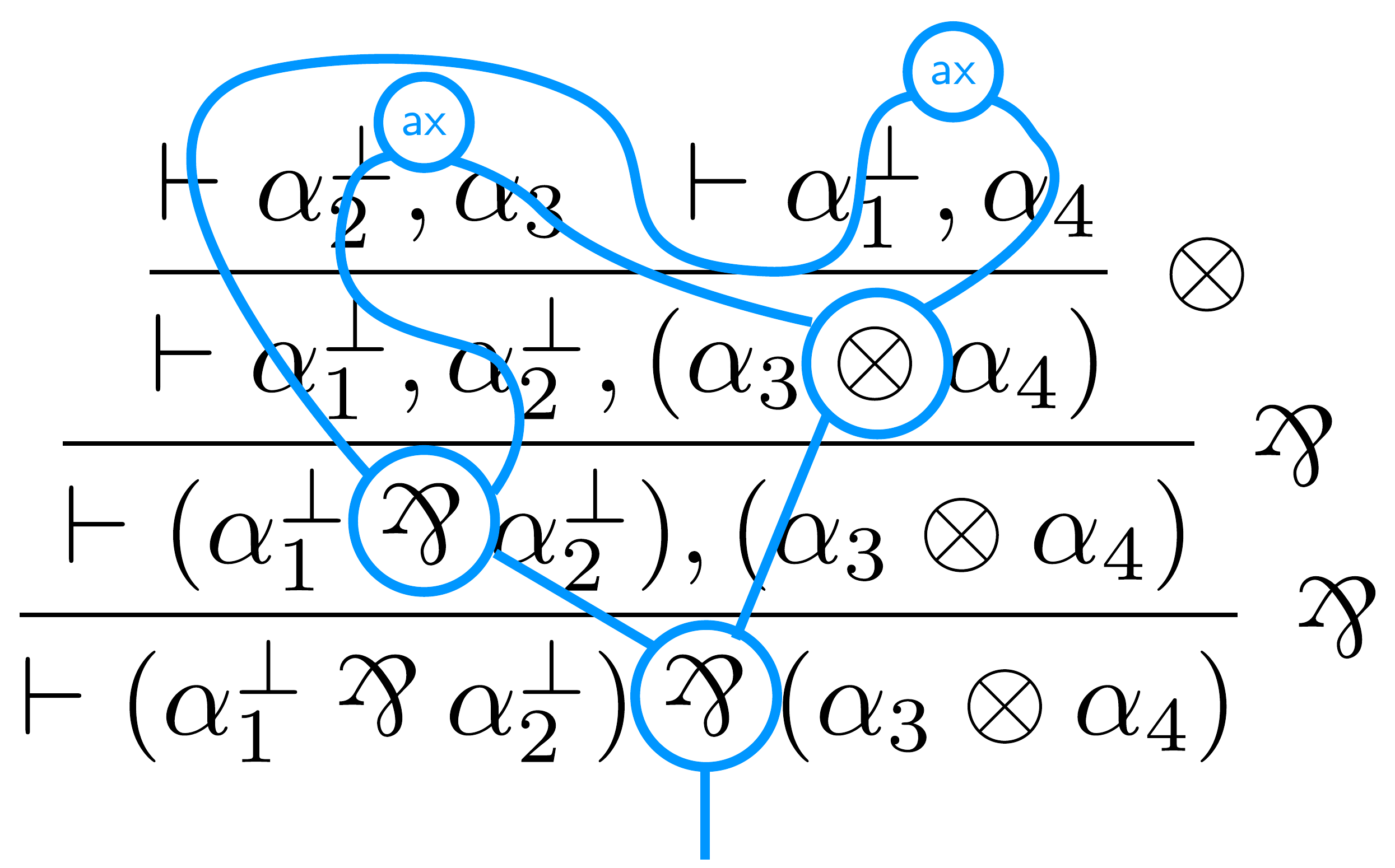}}\label{fig:pfnet_false}}}
\qquad\qquad
\raisebox{40pt}{\subfigure[]{\begin{tikzpicture}[ocenter]
%nodes
\node at (0,0) [etic](axlinkup){\tiny $\ax$};
\node at (axlinkup) [below = 2*\altax, etic](axlinkdown){\tiny $\ax$};
\node at (axlinkdown) [below left = 2*\altax and \ilar, etic](multlft){\tiny $\parr$};
\node at (axlinkdown) [below right = 2*\altax and \ilar, etic](multrht){\tiny $\otimes$};
\node at (multrht) [below left = 1.5*\altax and \ilar, etic](multcen){\tiny $\parr$};
\node at (multcen) [below = 1.5*\altax](concl){\tiny $B_\alpha$};
\node at (concl) [below= 1.5*\altax, etic](dummy){};
%edges
\draw[nopol, in=135, out=180, looseness=1.5] (axlinkup) to (multlft);
\draw[nopol, in=45, out=0] (axlinkup) to (multrht);
\draw[nopol, in=45, out=180, looseness=1.5] (axlinkdown) to (multlft);
\draw[nopol, in=135, out=0] (axlinkdown) to (multrht);
\draw[nopol, in=180, out=-90] (multlft) to (multcen);
\draw[nopol, in=0, out=-90] (multrht) to (multcen);
\draw[nopol, in=90, out=-90] (multcen) to (concl);
\end{tikzpicture}\label{fig:examples_false}}}
\qquad\qquad
\subfigure[]{\raisebox{7pt}{\includegraphics[scale=0.13]{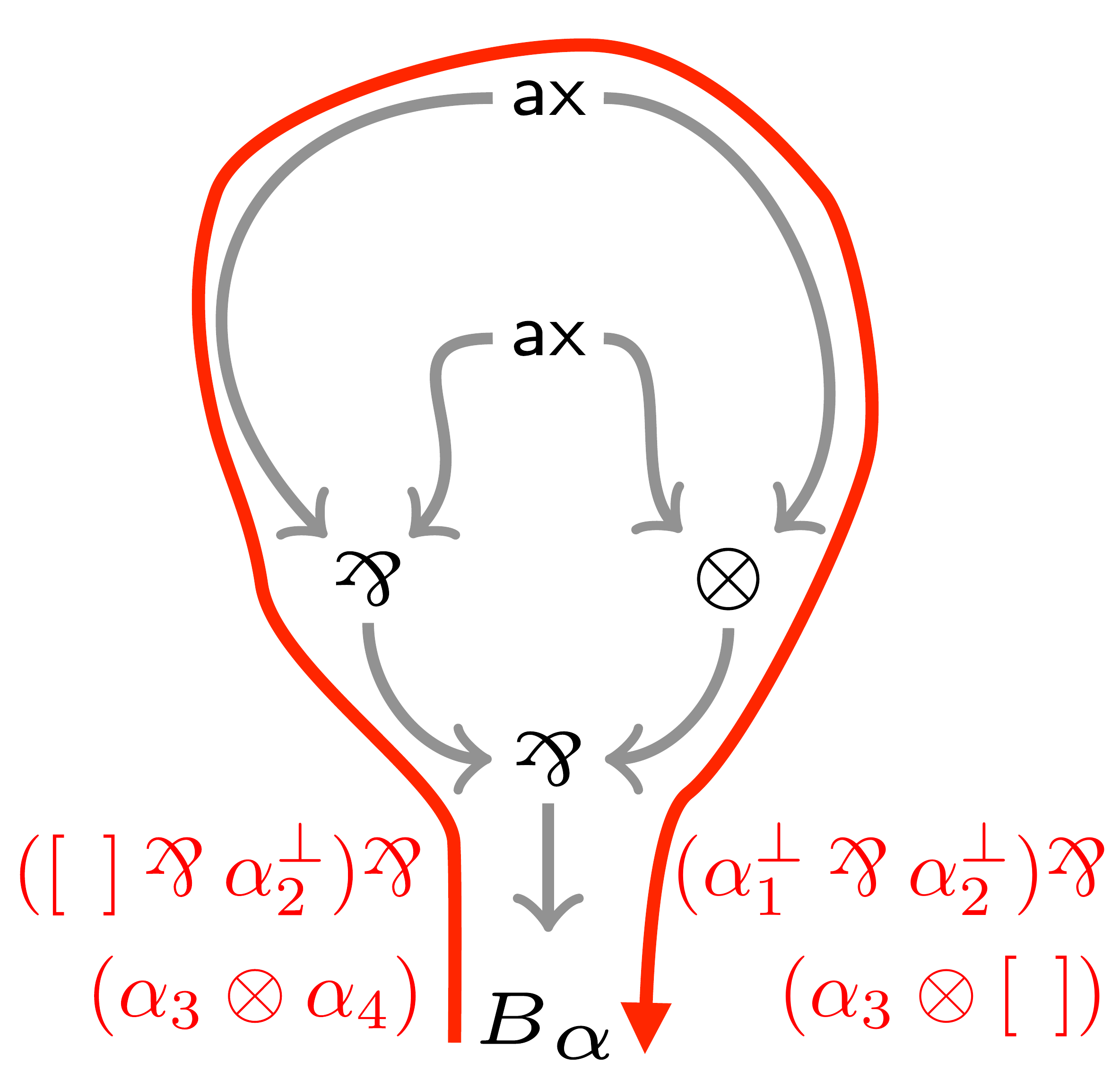}}\label{fig:pfnet_false_traceOne}}
\\ \vspace{6pt}
\subfigure[]{\raisebox{7pt}{\includegraphics[scale=0.13]{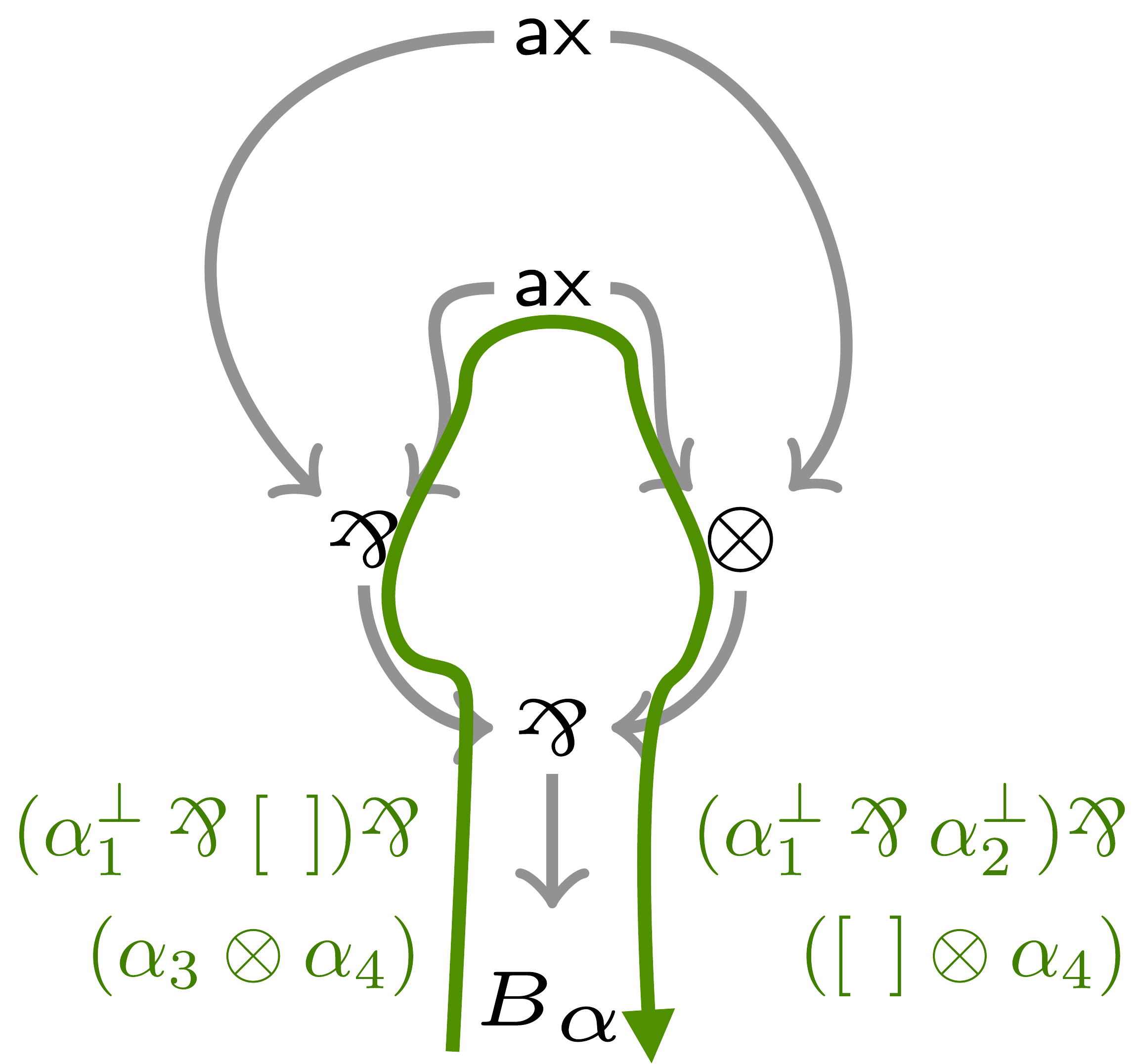}}\label{fig:pfnet_false_traceTwo}}
\qquad\qquad
\subfigure[]{\begin{tikzpicture}[ocenter]
%nodes
\node at (0,0) [etic](axlinklft){\tiny $\ax$};
\node at (axlinklft) [right = 2*\ilar, etic](axlinkrht){\tiny $\ax$};
\node at (axlinklft) [below = 2*\altax, etic](multlft){\tiny $\parr$};
\node at (axlinkrht) [below = 2*\altax, etic](multrht){\tiny $\otimes$};
\node at (multrht) [below left = 1.5*\altax and \ilar, etic](multcen){\tiny $\parr$};
\node at (multrht) [below= 3*\altax, etic](tensapp){\tiny $\otimes$};
\node at (multlft) [below= 4.5*\altax, etic](maincut){\tiny $\mathsf{cut}$};
\node at (multcen) [right= 2*\ilar, etic](axconcl){\tiny $\mathsf{ax}$};
\node at (axconcl) [below right= 2*\altax and \ilar, etic](concl){\tiny $B_\alpha$};
\node at (maincut) [above left= 1.5*\altax and 3*\ilar, etic](arglowpar){\tiny $\parr$};
\node at (arglowpar) [above= 2*\altax, etic](arguppar){\tiny $\parr$};
\node at (arguppar) [above left= 1.5*\altax and .7*\ilar, etic](argtens){\tiny $\otimes$};
\node at (argtens) [above left = 1.5*\altax and \ilar, etic](argaxlft){\tiny $\ax$};
\node at (argtens) [above right = 1.5*\altax and \ilar, etic](argaxrht){\tiny $\ax$};
\node at (maincut) [below= 1.5*\altax, etic](dummy){};
%edges
\draw[nopol, in=135, out=180, looseness=1.5] (axlinklft) to (multlft);
\draw[nopol, in=135, out=0] (axlinklft) to (multrht);
\draw[nopol, in=45, out=180, looseness=1.5] (axlinkrht) to (multlft);
\draw[nopol, in=45, out=0] (axlinkrht) to (multrht);
\draw[nopol, in=180, out=-90] (multlft) to (multcen);
\draw[nopol, in=0, out=-90] (multrht) to (multcen);
\draw[nopol, in=180, out=-90] (multcen) to (tensapp);
\draw[nopol, in=0, out=-90] (tensapp) to (maincut);
\draw[nopol, in=45, out=180] (axconcl) to (tensapp);
\draw[nopol, in=90, out=0] (axconcl) to (concl);
\draw[nopol, in=180, out=-90] (arglowpar) to (maincut);
\draw[nopol, in=45, out=-90, looseness=1.5] (arguppar) to (arglowpar);
\draw[nopol, in=135, out=-90, looseness=1.5] (argtens) to (arglowpar);
\draw[nopol, in=135, out=0] (argaxlft) to (argtens);
\draw[nopol, in=45, out=180] (argaxrht) to (argtens);
\draw[nopol, in=170, out=180, looseness=1.5] (argaxlft) to (arguppar);
\draw[nopol, in=45, out=0] (argaxrht) to (arguppar);
\end{tikzpicture}\label{fig:examples_nottrue}}
\\ \vspace{6pt}
\subfigure[]{\begin{tikzpicture}[ocenter]
%nodes
\node at (0,0) [etic](axlinklft){\tiny $\ax$};
\node at (axlinklft) [right = 2*\ilar, etic](axlinkrht){\tiny $\ax$};
\node at (axlinklft) [below = 2*\altax, etic](multlft){\tiny $\otimes$};
\node at (axlinkrht) [below = 2*\altax, etic](multrht){\tiny $\parr$};
\node at (multrht) [below left = 1.5*\altax and \ilar, etic](multcen){\tiny $\parr$};
\node at (multrht) [below right= 3*\altax and \ilar, etic](maincut){\tiny $\mathsf{cut}$};
\node at (maincut) [above right= 1.5*\altax and 2*\ilar, etic](arglowtens){\tiny $\otimes$};
\node at (arglowtens) [above right= 1.5*\altax and \ilar, etic](arguptens){\tiny $\otimes$};
\node at (arguptens) [above right= 1.5*\altax and 1.5*\ilar, etic](argax){\tiny $\ax$};
\node at (arglowtens) [above left= 1.5*\altax and \ilar, etic](argsucc){\tiny $\mathsf{succ}$};
\node at (arguptens) [above left= 1.5*\altax and \ilar, etic](argzero){\tiny $\mathsf{zero}$};
\node at (argax) [below right= 1.5*\altax and \ilar, etic](concl){\tiny $\mathbb{N}$};
\node at (maincut) [below= 1.5*\altax, etic](dummy){};
%edges
\draw[nopol, in=135, out=180, looseness=1.5] (axlinklft) to (multlft);
\draw[nopol, in=135, out=0] (axlinklft) to (multrht);
\draw[nopol, in=45, out=180, looseness=1.5] (axlinkrht) to (multlft);
\draw[nopol, in=45, out=0] (axlinkrht) to (multrht);
\draw[nopol, in=180, out=-90] (multlft) to (multcen);
\draw[nopol, in=0, out=-90] (multrht) to (multcen);
\draw[nopol, in=180, out=-90] (multcen) to (maincut);
\draw[nopol, in=0, out=-90] (arglowtens) to (maincut);
\draw[nopol, in=134, out=-90] (argsucc) to (arglowtens);
\draw[nopol, in=45, out=-90] (arguptens) to (arglowtens);
\draw[nopol, in=135, out=-90] (argzero) to (arguptens);
\draw[nopol, in=45, out=-180] (argax) to (arguptens);
\draw[nopol, in=90, out=0] (argax) to (concl);
\end{tikzpicture}\label{fig:examples_pcf}}
\qquad\qquad
\subfigure[]{\begin{tikzpicture}[ocenter]
%nodes
\node at (0,0) [etic](axlinklft){\tiny $\ax$};
\node at (axlinklft) [right = 2*\ilar, etic](axlinkrht){\tiny $\ax$};
\node at (axlinkrht) [below left = 1.5*\altax and .5*\ilar, etic](hada){\tiny $\mathsf{H}$};
\node at (axlinklft) [below = 3*\altax, etic](multlft){\tiny $\parr$};
\node at (axlinkrht) [below = 3*\altax, etic](multrht){\tiny $\otimes$};
\node at (multrht) [below = 1.25*\altax, etic](cnot){\tiny $\mathsf{CNOT}$};
\node at (multrht) [below left = 2.25*\altax and \ilar, etic](multcen){\tiny $\parr$};
\node at (multrht) [below right= 4.5*\altax and \ilar, etic](maincut){\tiny $\mathsf{cut}$};
\node at (maincut) [above right= 1.5*\altax and 3*\ilar, etic](arglowtens){\tiny $\otimes$};
\node at (arglowtens) [above left= 1.5*\altax and \ilar, etic](arguptens){\tiny $\otimes$};
\node at (arglowtens) [above right= 1.5*\altax and 1.5*\ilar, etic](argax){\tiny $\ax$};
\node at (arguptens) [above left= 1.5*\altax and \ilar, etic](argzerolft){\tiny $\mathsf{zero}$};
\node at (arguptens) [above right= 1.5*\altax and \ilar, etic](argzerorht){\tiny $\mathsf{zero}$};
\node at (argax) [below right= 1.5*\altax and \ilar, etic](concl){\tiny $\mathbb{Q}\tens\mathbb{Q}$};
\node at (maincut) [below= 1.5*\altax, etic](dummy){};
%edges
\draw[nopol, in=135, out=180, looseness=1.5] (axlinklft) to (multlft);
\draw[nopol, in=90, out=0] (axlinklft) to (hada);
\draw[nopol, in=135, out=-90] (hada) to (multrht);
\draw[nopol, in=45, out=180, looseness=1.5] (axlinkrht) to (multlft);
\draw[nopol, in=45, out=0] (axlinkrht) to (multrht);
\draw[nopol, in=180, out=-90] (multlft) to (multcen);
\draw[nopol, in=90, out=-90] (multrht) to (cnot);
\draw[nopol, in=0, out=-90] (cnot) to (multcen);
\draw[nopol, in=180, out=-90] (multcen) to (maincut);
\draw[nopol, in=0, out=-90] (arglowtens) to (maincut);
\draw[nopol, in=135, out=-90] (argzerolft) to (arguptens);
\draw[nopol, in=45, out=-90] (argzerorht) to (arguptens);
\draw[nopol, in=135, out=-90] (arguptens) to (arglowtens);
\draw[nopol, in=45, out=-180] (argax) to (arglowtens);
\draw[nopol, in=90, out=0] (argax) to (concl);
\end{tikzpicture}\label{fig:examples_quantum}}
\end{center}
\end{minipage}}
\caption{Nets and Abstract Machines --- Some Examples}\label{fig:examples_mll}
\end{center}
\end{figure*}
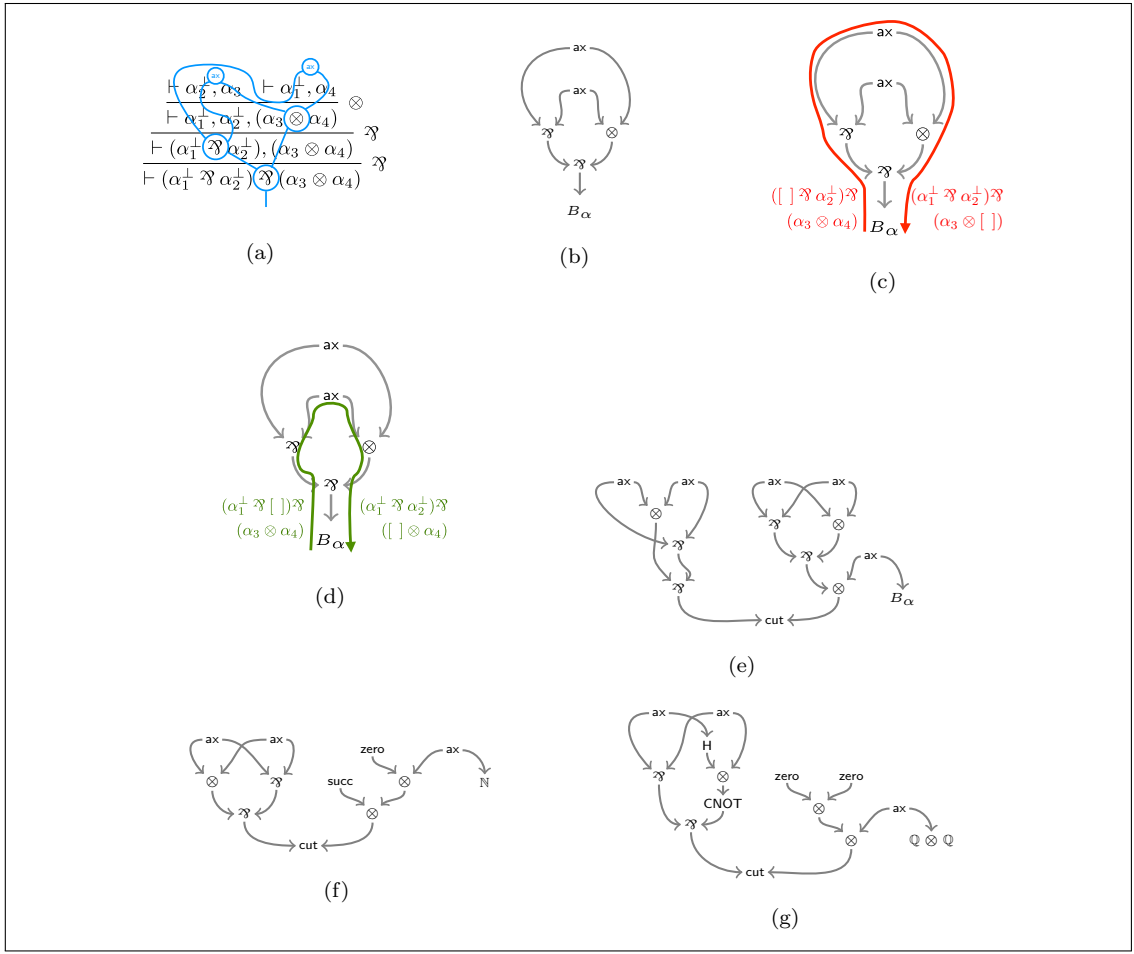

The ideas above have been extensively developed. In particular,
Interaction Abstract Machines have been defined for $\lambda$-calculi
in which duplication is indeed possible, and also for
\emph{applicative} $\lambda$-calculi, i.e., calculi endowed with
constants and possibly recursion. Take, as an example, the simple
\PCF\ term
$\termone=(\abstr{\varone}{\abstr{\vartwo}{\varone\vartwo}})\;\suc\;\zero$
where $\suc\colon\mathbb{N}\multimap\mathbb{N}$ and
$\zero\colon\mathbb{N}$ are constants for successor and zero,
respectively. The fact that the term above evaluates to $\one$, again,
can be observed by letting \emph{one} token travel inside the
(proof-net corresponding to a) type derivation for $\termone$, as
depicted in Figure~\ref{fig:examples_pcf}. The token now starts its
journey from the node labeled with $\mathsf{zero}$ carrying the
natural number ``built so far'', which initially is of course
$0$. After some re-routing induced by the multiplicative nodes $\parr$
and $\tens$, the token reaches the node labelled with $\mathsf{succ}$,
which modifies the natural number to $\one$, and the journey proceeds
until the conclusion.

How about quantum computation? Would it be possible to adapt the
scheme above to $\lambda$-calculi specifically designed for quantum
computation? Token machines seem to be a natural way to model
inherently linear calculi such as quantum $\lambda$-calculi. However,
if one tries to \emph{directly} apply the paradigm described above,
one soon gets into troubles. Consider, as an example, a term like
$\termtwo=(\abstr{\pair{\varone}{\vartwo}}{\cnot\pair{\hada{\varone}}{\vartwo}})\pair{\zero}{\zero}$,
where $\cnot$ and $\hada$ are certain unitary operations that act on
2- and 1-qubit systems, respectively.  This is an encoding of a
quantum circuit having the remarkable property of producing an
entangled pair of qubits in output. Let us try to play the same game
we played with $\termone$ on a graph-theoretic representation of
$\termtwo$ (see Figure~\ref{fig:examples_quantum}). If we allow a
token to start its journey from the leftmost occurrence of $\zero$, it
can of course reach $\hada$, go through it (having the underlying
qubit modified accordingly) but gets stuck at $\cnot$. Indeed, the
value of the first and second outputs of $\cnot$ can only be known
when \emph{both} inputs are available. But even more importantly, the
state of those qubits is an entangled state, i.e. it cannot be
described as the tensor product of the two qubits. Switching to a
setting in which an IAM state consists not of a single token but
possibly of multiple ones seems very natural now. Moreover, as is
clear from the $\cnot$ example, there should be a way to force those
many tokens to \emph{synchronize}, i.e., to wait until some of the
other tokens reaches a certain state, before proceeding.

In all the ``non quantum'' examples, \emph{multiple} tokens could travel the
net in parallel, but what we compute is exactly the same, because
the tokens do not interact with each other. Things are very different
in the quantum example. Summing up, as computational models the GoI
machines are effective, powerful (as shown by the results on optimal
reductions and implicit complexity from the literature). However, 
they have  limits:
\begin{varitemize}
\item 
  a \emph{general} limit in expressiveness, as they capture
  parallel computation, but without synchronization;
\item 
  a \emph{specific} limit, as it is not possible to model quantum
  computing without some form of synchronization.
\end{varitemize}
%%%%%%%%%%%%%%%%%%%%%%%%%%%%%%%%%%%%%%%%
\paragraph{Synchronizations and Deadlocks.}
%%%%%%%%%%%%%%%%%%%%%%%%%%%%%%%%%%%%%%%% 
As the reader may expect, the most delicate property in a multitoken
setting is deadlock freedom.  Consider Figure
\ref{fig_nodeadlock}. Squares connected with lines represent
synchronization points: tokens should cross simultaneously $s_1$ and
cross simultaneously $s_2$. In this configuration, if we have a token
on each of the three positions $\postwo_1$,$\postwo_2$,$\postwo_3$,
they are in deadlock.  In the following, we develop an approach which
reduces the absence of deadlocks in the machine $M_R$ interpreting a
proof-net $R$ to the correctness criterion of $R$, our \emph{motto}
then being the following:
\begin{center}
``Correctness of $R$'' $\quad\Rightarrow\quad$ ``Deadlock Freedom of $M_R$''.
\end{center}
\begin{figure}[htbp]
\begin{center}
  \fbox{
  \begin{minipage}{\figurewidth}
  \centering
  \vspace{4pt}
  \begin{tikzpicture}[ocenter]
\node at (0,0) [etic](sylftd){};
\node at (sylftd) [etic, above= 0.8*\stalt] (sylftcc){};
\node at (sylftcc) [etic, above= 0.02*\stalt] (sylftccaux){$\postwo_1$};
\lsync{sylftd}{sylftcc}{sylftcd};
\node at (sylftcd) [etic, below left = 0.2*\stalt and 0.3*\stalt, etic] (sylftcdlabel){$s_1$};
\node at (sylftcc) [etic, above=1.3*\stalt] (sylftuu){};
\node at (sylftd) [right= 4*\altax, etic](syctrd){};
\node at (syctrd) [etic, above= 0.8*\stalt] (syctrcc){};
\lsync{syctrd}{syctrcc}{syctrcd};
\node at (syctrcd) [etic, above= 0.8*\stalt] (syctru){};
\node at (syctru) [etic, above= 0.02*\stalt] (syctruaux){$\postwo_2$};
\lsync{syctrcd}{syctru}{syctrcu};
\node at (syctrcc) [etic, above=\stalt] (syctruu){};
\node at (syctrcd) [right=4*\altax, etic](syrhtcd){};
\node at (syrhtcd) [etic, above= 0.8*\stalt] (syrhtcu){};
\lsync{syrhtcd}{syrhtcu}{syrhtcc};
\node at (syrhtcc) [etic, right = 0.4*\stalt, etic] (syrhtcclabel){$s_2$};
\node at (syrhtcc) [etic, above= 0.8*\stalt] (syrhtu){};
\node at (syrhtu) [etic, above= 0.02*\stalt] (syrhtuaux){$\postwo_3$};
\lsync{syrhtcc}{syrhtu}{syrhtcu};
\node at (syrhtcc) [etic, above=1.3*\stalt] (syrhtuu){};
\draw[thick] (sylftcd) to (syctrcd);
\draw[thick] (syrhtcc) to (syctrcc);
\draw[thick] (syrhtcu) to[out=0,in=0] (syrhtuu) to[out=180,in=0] (sylftuu) to[out=180,in=180] (sylftcd);
\end{tikzpicture}
  \vspace{4pt}
  \end{minipage}}
  \condnocaptionrule\caption{Deadlocked Structures}\label{fig_nodeadlock}
\end{center}
\end{figure}
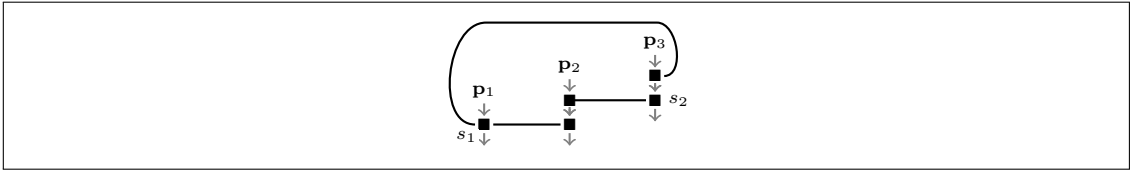

%%%%%%%%%%%%%%%%%%%%%%%%%%%%%
\paragraph{Some Related Work.}
%%%%%%%%%%%%%%%%%%%%%%%%%%%%%
The initial motivation leading us to the development of Geometry of
Synchronization was the study of quantum computation.

Multiple-token machines have already been investigated by the first
author and Margherita Zorzi in a recent unpublished
manuscript~\cite{DalLagoZorzi}. In this paper, we take a step back and 
analyze the construction from a logical point of view, showing how
parallelism and synchronization can be satisfactorily captured within
a slight variation of multiplicative linear logic.

A version of quantum proof-nets have recently been proposed by the
authors~\cite{esop}. Boxes are seen as a way to implement quantum
measurement, and the reader can find several useful examples
there. The proposed class of proof-nets has weaker properties,
however: the results on cut elimination are rather limited, and do not
allow us to study deadlock freedom as we do here. Having better proof
theoretical properties, as in the nets in this paper, allows us to go
further in the interpretation of quantum lambda calculi, which we
are able to simulate in a sound way. Moreover, we uncover and separate
the classic computational structure, making it independent from
quantum data. Not only this allows to modularize the results, but we
believe that other interesting application examples for the \SIAM,
besides quantum computing, can be found, such as distributed
implementations (see, e.g.~\cite{DBLP:conf/lics/FredrikssonG13}).

We finally like to mention other work which is related to \SMLL\ in that the aim is to capture  synchronization into a logical calculus, in particular  \cite{hirai} and \cite{DK00}.

\newcommand{\rednetP}{\rednet_2}

\newcommand{\redsiam}{\longrightarrow}

\newcommand{\tok}{\mathbf{T}}
\newcommand{\toktwo}{\mathbf{U}}
\newcommand{\Itok}{\mathbf I}
\newcommand{\Ftok}{\mathbf F}
\newcommand{\pp}{\mathbf p}

\renewcommand{\ss}{\mathbf s}

\newcommand{\clean}[1]{(#1)^\circle}

\newcommand{\ie}{\emph{i.e.}}
  
%%%%%%%%%%%%%%%%%%
\section{\SMLLb}
%%%%%%%%%%%%%%%%%
In this section we introduce \SMLLb\ nets, which are a generalization
of proof-nets for Multiplicative Linear Logic (\MLL). As a reference
to proof-nets, we suggest \cite{LaurentTorino} --- our approach to
correctness is close to the one described there.

\noindent\emph{Formulas.} The language of \SMLL\ \emph{formulas} is identical
to the one for \MLL, i.e.,
$$ 
\formone\bnf\one\midd\bot\midd X\midd X\b\midd\formone\otimes\formone\midd\formone\parr\formone,
$$
where $X$ ranges over a denumerable set of \emph{propositional
  variables}. The constants $1,\bot$ are the \emph{units}.  We call
\emph{atomic} those formulas which are either propositional variables
or units. %We will indicate atoms by $\alpha,\gamma$ etc.
Linear negation $(\cdot)\b$ is extended into an involution on all
formulas as usual: $X\equiv X\b\b$, $1\b\equiv \bot $, $\bot\b\equiv 1
$, $(\formone\otimes\formtwo)\b\equiv \formone\b\parr \formtwo\b $,
$(\formone\parr\formtwo)\b\equiv \formone\b\otimes \formtwo\b$.
Linear implication is a defined connective:
$\formone\lin\formtwo\equiv\formone\b\parr\formtwo$.
  
\noindent\emph{Polarized Formulas.} Atoms and connectives of
\MLL\ are divided in two classes: positive ($1, X, \otimes $) and
negative ($\bot, X\b, \parr $).  In this paper, to have a compact
presentation, we exclude from polarized formulas the propositional
variables, and define \emph{positive} formulas (denoted by $P$) and
\emph{negative} formulas (denoted by $N$) as follows: 
\begin{align*}
  \posfone&\bnf \one\midd\posfone\otimes\posfone;\\
  \negfone&\bnf \bot\midd\negfone\parr\negfone.
\end{align*}
Observe that the formula $\formone=\one\lin\one\equiv \bot\parr\one$
contains one negative occurrence of atom and one positive occurrence
of atom; thus $\formone$, in our setting, is neither positive nor
negative.

%%%%%%%%%%%%%%%%%%%%%%%%
\subsection{Structures}
%%%%%%%%%%%%%%%%%%%%%%%%
An \SMLLb\ \emph{structure} is a labeled directed (multi-)graph, where
the edges are labeled with \MLL\ formulas (the label of an edge is
called its \emph{type}). The alphabet of \emph{nodes}, given in Figure
\ref{sMLLlinks} (a,b,c), is the same as the one of \MLL, but extended
with a new link, called a \emph{sync link}. Altogether, we have \MLL\
links $\{\axlk, \cut, \otimes, \parr \}$, unit links
$\{\onelk,\botlk\}$, and sync links (Figure~\ref{fig:sync_links}),
which we detail below.

Graphically, we represent structures with the edges oriented from top
to bottom; we use accordingly terms like ``above'', ``below'', ``upwards'' and
``downwards''.  We call \emph{conclusions} (resp. \emph{premisses}) of a
link those edges represented below (resp. above) the link symbol.
\begin{figure}[htbp]
  \begin{center}
    \fbox{
      \centering
      \begin{minipage}{\figurewidth}
        \vspace{4pt}
\begin{center}
\subfigure[]{
% AXIOM
\begin{tikzpicture}[ocenter]
\node at (0,0) [label=below:$A$, etic](axposedge){};
\node at (axposedge.center) [above right = \altax and \ilar, etic](axlink){\tiny $\ax$};
\node at (axlink.center) [below right= \altax and \ilar,  etic](axnegedge){};
\node at (axnegedge.center) [etic, label = below: $A^\bot$, above right=2pt and 3pt](dummy){};
% edges
\draw[nopolrev, in=180, out=90] (axposedge) to (axlink);
\draw[nopol, in=90, out=0] (axlink) to (axnegedge);
\end{tikzpicture}
\qquad
% CUT
\begin{tikzpicture}[ocenter]
\node at (0,0) [label=above:$A$, etic](axposedge){};
\node at (axposedge.center) [below right = \altax and \ilar, etic](axlink){\tiny $\cut$};
\node at (axlink.center) [above right= \altax and \ilar, label=above:$A^\bot$, etic](axnegedge){};
% edges
\draw[nopolrev, in=-90, out=180] (axlink) to (axposedge);
\draw[nopol, in=0, out=-90] (axnegedge) to (axlink);
\end{tikzpicture}
\qquad
% TENSOR
\begin{tikzpicture}[ocenter]
\node at (0,0)[etic] (tensPal){\tiny $A\tens B$};
\node at (tensPal.center) [etic,above left = \stalt and \hstlar](tensLPax){\tiny $A$};
\node at (tensPal.center) [etic,above right = \stalt and \hstlar,label distance = 0.5pt](tensRPax){\tiny $B$};
\ltens{tensPal}{tensLPax}{tensRPax}{tensor};
\end{tikzpicture}
\qquad
% PAR
\begin{tikzpicture}[ocenter]
\node at (0,0)[etic] (parrPal){\tiny $A\parr B$};
\node at (parrPal.center) [etic,above left = \stalt and \hstlar](parrLPax){\tiny $A$};
\node at (parrPal.center) [etic,above right = \stalt and \hstlar,label distance = 0.5pt](parrRPax){\tiny $B$};
\lpar{parrPal}{parrLPax}{parrRPax}{par};
\end{tikzpicture}\label{fig:mll_links}}
\\ \vspace{5pt}
\subfigure[]{
% SYNC
\begin{tikzpicture}[ocenter]
\node at (0,0) [etic](syplftPal){\tiny $P_1$};
\node at (syplftPal.center) [etic, above= \stalt] (syplftPax){\tiny $P_1$};
\lsync{syplftPal}{syplftPax}{syplft};
\node at (syplftPal) [right= 3*\altax, etic](syprhtPal){\tiny $P_k$};
\node at (syprhtPal.center) [etic, above= \stalt] (syprhtPax){\tiny $P_k$};
\lsync{syprhtPal}{syprhtPax}{syprht};
\
\node at (syprhtPal) [right=2*\altax, etic](synlftPal){\tiny $N_1$};
\node at (synlftPal.center) [etic, above= \stalt] (synlftPax){\tiny $N_1$};
\lsync{synlftPal}{synlftPax}{synlft};
\node at (synlftPal) [right= 3*\altax, etic](synrhtPal){\tiny $N_m$};
\node at (synrhtPal.center) [etic, above= \stalt] (synrhtPax){\tiny $N_m$};
\lsync{synrhtPal}{synrhtPax}{synrht};
% dots
\node at (syplftPal.center) [right=1.5*\altax, etic](uldots){\tiny $\cdots$};
\node at (syplftPax.center) [right=1.5*\altax, etic](lldots){\tiny $\cdots$};
\node at (synlftPal.center) [right=1.5*\altax, etic](urdots){\tiny $\cdots$};
\node at (synlftPax.center) [right=1.5*\altax, etic](lrdots){\tiny $\cdots$};
% links
\draw[thick] (syplft) to (syprht);
\draw[thick] (syprht) to (synlft);
\draw[thick] (synlft) to (synrht);
\end{tikzpicture}\label{fig:sync_links}}
\qquad
\subfigure[]{
% ONE
\begin{tikzpicture}[ocenter]
\node at (0,0) [etic](onePal){\tiny $1$};
\lone{onePal}{one};
\end{tikzpicture}
\quad
% BOT
\begin{tikzpicture}[ocenter]
\node at (0,0) [etic](botPal){\tiny $\bot$};
\lbot{botPal}{bot};
\end{tikzpicture}\label{fig:unit_links}}
\qquad
\subfigure[]{
\begin{tikzpicture}[ocenter]
% bang
\node at (0,0) [etic](pospal){};
\abox{pospal}{exbox}{\stboxlw}{\stboxrwaux-10pt}{\stboxh+6pt}
\node at (pospal.center) [etic](possym){$\bot$};
% auxiliary ports
\node at (pospal.center)[right=\stlar, etic](posauxkport){\tiny $\Gamma$};
\node at (posauxkport)[below=\hstalt*1.2, etic](posauxkportcon){\tiny $\Gamma$};
% net g
\node at (pospal.center)[above right=0.8*\stalt and \stlar, net](netg){$R$};
\draw[nopolgen, in=90, out=-90](netg)to(posauxkport);
\draw[nopolgen, in=90, out=-90](posauxkport)to(posauxkportcon);
% bot node
\node at (possym.center) [above=1.4*\altax, etic](botlink){\tiny $\mathsf{bot}$};
\draw[nopol, in=90, out=-90] (botlink) to (possym);
% principal port
\node at (pospal.center)[below = \hstalt*1.2, etic](posprincport){\tiny $\bot$};
\draw[nopol](possym)to(posprincport);
\end{tikzpicture}\label{fig:box_links}}
\end{center}
      \end{minipage}}
    \condnocaptionrule\caption{\SMLL\ Links}\label{sMLLlinks}  % (a.)= \MLL\, (a.)+(b.)=\SMLLcf\,(a.)+(b.)+(c.)=\SMLL\}
  \end{center}
\end{figure}
The sort of a link induces constraints on the number and the labels of
its premisses and conclusions, as shown in Figure \ref{sMLLlinks}. The
graph can have pending edges, \ie, some edges may not have a target;
the pending edges are called the \emph{conclusions} of the
structure. We will often say that a link ``has a conclusion (premiss)
$\formone$'' as shortcut for ``has a conclusion (premiss) of type
$\formone$''.  When we need more precision, we distinguish between an
edge and its type, and we use variables such as $e,f$ for the edges.

\noindent\emph{Sync links.}
A sync link has $n$ premisses, and $n$ conclusions. For each $i$
($1\leq i\leq n$) the $i$-th premiss $e_i$ and the
\emph{corresponding} $i$-th conclusion $f_i$ are typed by the
\emph{same} formula, which is either positive or negative.  To stress
the correspondence between the $i$-th premiss and the $i$-th
conclusion, we find it convenient to graphically represent an $n$-ary
sync link as a list of $n$ nodes connected by untyped edges.  In the
example in Figure~\ref{fig:sync_links} , we have $n=k+m$ edges, which
are typed with $k$ positive and $m$ negative formulas.

We now need some specific terminology. An edge is \emph{positive}
(\emph{negative}) if its type is positive (or negative); we also say
that such an edge is \emph{polarized}. A node is polarized if its
conclusions are all polarized.  All edges of sync links are polarized;
we will borrow some of the terminology from polarized linear logic
\cite{phdlaurent}.  Given a sync link, we call \emph{in-edges} its
positive premisses and negative conclusions, and call \emph{out-edges}
the positive conclusion and negative premisses.

Intuitively, a sync link acts on an edge of type $A$, but does not
introduce $A$.  We call \emph{sync path} a path which traverses only
sync links, going in and coming out on corresponding edges.  We say
that an edge $e$ of type $A$ is a \textit{hereditary conclusion} of a
link $l$ if there is a sync path from the conclusion of $l$ to $e$ .
 
\noindent\emph{Units and unit-free fragment.} The units are $\one$ and
$\bot$, respectively introduced by the links $\onelk$ and $\boxlk$. To
the $\botlk$ link is associated a notion of box
(Figure~\ref{fig:box_links}) which we discuss next.
 
We indicate by \SMLLcf\ the fragment of \SMLLb\ without unit links
(therefore in particular without boxes); in \SMLLcf, the formulas $1$
and $\bot$ are hence only introduced as conclusions of axioms. Even if
minimal, \SMLLcf\ is actually an interesting system in itself, and
especially well behaved; we will study its specific properties.
 
%%%%%%%%%%%%%%%%%%%%%%%%%%%%%%%
\paragraph{Structures with Boxes.}
%%%%%%%%%%%%%%%%%%%%%%%%%%%%%%%
We use boxes to represent  the rule for $\bot$:
$$
\infer{\vdash \bot, \Gamma}{\vdash \Gamma}
$$
The definition of structures with boxes which we adopt is standard.
\emph{In short}, a box of conclusions $\bot,\Gamma$ contains a
structure $R$ of conclusions $\Gamma$, and a distinguished $\botlk$
link of conclusion $\bot$.  Such a conclusion $\bot$ is the
\emph{lock} of the box, and $R$ is its \emph{content}.  We represent a box
graphically as in Figure \ref{sMLLlinks}(d), where the structure $R$
is represented as a circle inside the box, and the barred edge
labelled by $\Gamma$ stands for a sequence of edges (the conclusions
$\Gamma$).

More formally, an \SMLLb\ \emph{structure with boxes} is an \SMLLb\
structure together with a function which associates to each node $l$
of sort $\botlk$ a sub-structure $R$ of conclusions $\Gamma$  (as mentioned above, $l$ and $R$ are depicted in a box).  Boxes
are required to be either one included in the other, or disjoint. The
\emph{depth} of a node is the number of boxes to which it belongs,
while the depth of a structure is the maximal depth of its nodes. The
lock of a box acts as a guard, as will be evident when we define
normalization and the \SIAM.
%%%%%%%%%%%%%%%%%%%%%%%%
\subsection{Correctness}
%%%%%%%%%%%%%%%%%%%%%%%%
A \emph{net} is a structure (with boxes) which fulfills a correctness
criterion. We define correctness by means of {switching paths} (see
\cite{LaurentTorino}). A \emph{switching path} on the structure $R$ is
an undirected path\footnote{By path, in this paper we always mean a
  \emph{simple path} (no repetition of either nodes or edges ).}
which uses:
\begin{varitemize}
\item 
  for each $\parr$ link,  at most one of the two premisses;
\item 
  for each  sync link, at most one of its \emph{out-edges}. 
\end{varitemize} 
The former condition is standard, the  latter condition rules out paths such as the one 
 going from $P_1$ (below)
  to $N_1$ (above) or to $P_K$ (below) in Figure~\ref{sMLLlinks}(b).
  
Let us first state correctness for  \SMLLcf, as it is as immediate as for \MLL:
an \SMLLcf\ structure  is  \emph{correct}  if none of its switching paths is cyclic.
Correctness for an \SMLLb\ structure $R$ is defined by levels, as usual with boxes.

We call \emph{0-graph} of $R$ the restriction to depth $0$ of the graph
which is obtained from $R$ by replacing each box of conclusion $\bot,
\Gamma$ with a new sort of node, labelled as $\boxlk$, which has the
same conclusions $\bot, \Gamma$ (like $\axlk$, a $\boxlk$ node has no
premisses). An \SMLLb\ structure with boxes is correct --- and is said
to be an \SMLLb\ \emph{net} --- if the following conditions hold:
\begin{varenumerate}
\item  
 there is no switching cycle in the $0$-graph of $R$;
\item  
  the structure inside each box is itself correct.
\end{varenumerate}
It is immediate to verify that if we only consider \MLL\ links, we
simply have a formulation in terms of switching paths of the usual
``acyclicity condition'' in the Danos-Regnier criterion\footnote{The
  Danos-Regnier criterion~\cite{DanosRegnierMult} is actually made of
  two conditions namely ``acyclic" and "connected"; however
  connectedness only role is to rule out the ``mix'' rule from the
  sequent calculus. This is not relevant in our development, so we
  will ignore it (if wished, one can introduce in the standard way
  also a connectedness condition; we would then speak of
  \emph{connected} nets).}.

The correctness criterion is a key ingredient to guarantee that
synchronizations behave well, i.e. that there are no deadlocks,
neither in the normalization nor in the \SIAM\ machine.
%%%%%%%%%%%%%%%%%%%%%%%%%%%%%%%
\paragraph{Absence of Deadlocks.}
%%%%%%%%%%%%%%%%%%%%%%%%%%%%%%%
As already discussed, the central issue associated to the introduction
of synchronizations is the need to guarantee the absence of deadlocks,
both in the normalization of the nets, and in the runs of the \SIAM.
We now introduce some technical notions and give a lemma which will be
our main tool in all proofs of deadlock freedom.  It is common to
verify deadlock freedom by using a notion of strict partial order, and
this is the case also in our setting.  More precisely, we define a
partial order on the sync links; the order corresponds to a notion of
dependency that will become clear when we define the \SIAM\
machine. We prove that the order is a strict partial order; this
indicates that there is always at least one sync link which does not
depend on any other one.  
\renewcommand{\before}{\prec}
\newcommand{\linkone}{l}

Given two links $\linkone_1, \linkone_2$ of an \SMLLb\ net, we write
$\linkone_1\before\linkone_2$ (and we say that $\linkone_1$ is
\emph{before} $\linkone_2$) if there is a \emph{polarized path} from
$\linkone_1$ to $\linkone_2$, i.e., a path of polarized edges
(connecting polarized nodes) which is going \emph{upwards on negative}
edges, and \emph{downwards on positive} edges. We ask that a polarized
path does not enter boxes.

%\condmedskip
\begin{lemma}[Links Strict Order]\label{no_deadlocks} 
 Given a net $R$, the  set of its links equipped with the relation $\before$ is a finite strict partial order.
\end{lemma} 
%\condmedskip 
The result follows from the fact that a polarized path
$p$ is in particular a switching path (as one can easily check,
noticing that if $p$ crosses a sync link, it uses at most one
out-edge); hence a polarized path is never cyclic, and the relation is
irreflexive. As a consequence, configurations like the ones in Figure
\ref{fig_nodeadlock} are not possible. 
\begin{proof}  
  The relation is transitive by construction. To prove that it is
  irreflexive, we show that a polarized path cannot be cyclic, by
  proving that a polarized path $p$ is in particular a switching path.
  If $p$ crosses a sync link, it uses at most one out-edges.  If $p$
  crosses a par link $l$, we know that $l$ is a polarized node, which
  implies that both its premisses are negative; therefore $p$ uses at
  most one of the premisses.
\end{proof}
%%%%%%%%%%%%%%%%%%%%%%%%%%%%
\subsection{Normalization}
%%%%%%%%%%%%%%%%%%%%%%%%%%%%
We define a set of rewriting rules on \SMLLb\ nets.  The elementary
reduction steps are given in Figure \ref{sync_red}. Reduction is
intended to happen at level $0$, i.e. reduction \emph{cannot} take place
inside a box. This way we obtain a rewrite relation on structures, called
$\rightarrow$.
\begin{figure}[htbp]
  \begin{center}
  \input{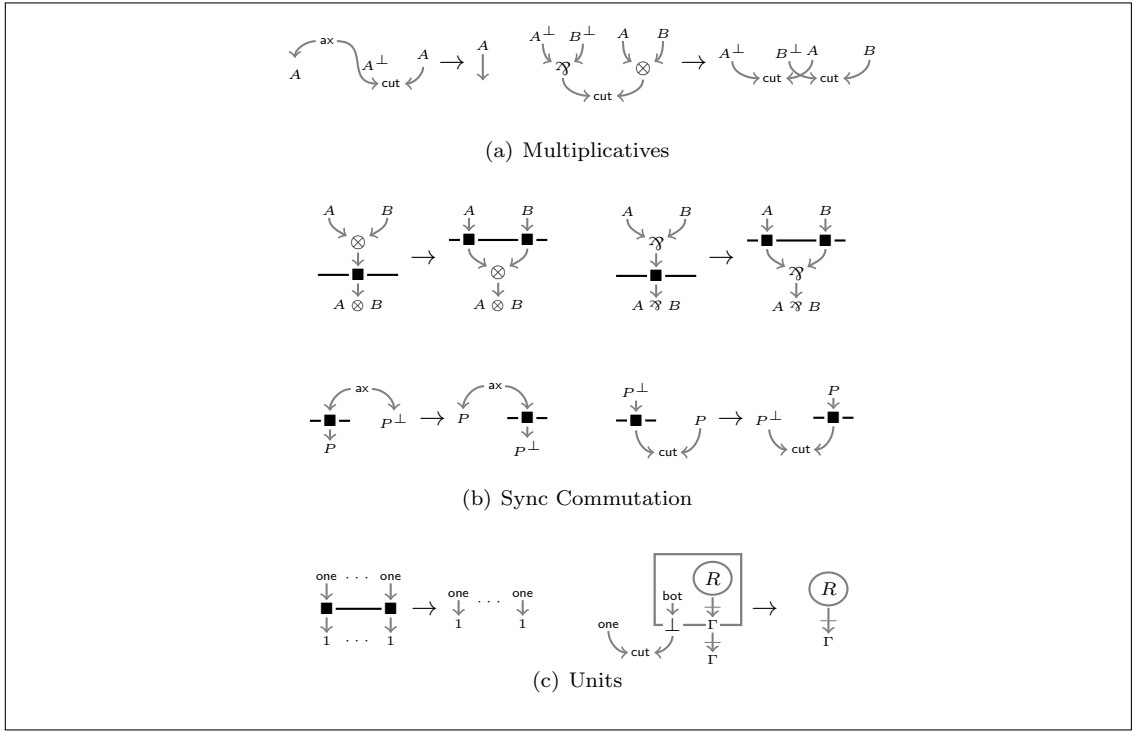}
  \condnocaptionrule\caption{Reduction}\label{sync_red}
  \end{center}
\end{figure}
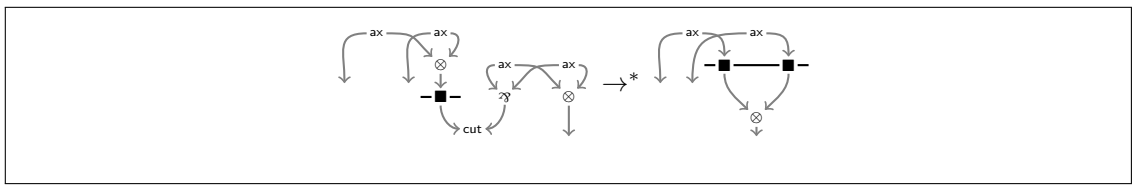
\begin{figure}[htbp]
  \begin{center}
  \fbox{
\begin{minipage}{\figurewidth}
\vspace{5pt}
\begin{center}
\begin{tikzpicture}[ocenter]
%nodes
\node at (0,0) [etic](axlinklftlft){\tiny $\ax$};
\node at (axlinklftlft) [right = 2*\ilar, etic](axlinklftrht){\tiny $\ax$};
\node at (axlinklftrht) [below right = 1.5*\altax and 2*\ilar, etic](axlinkrhtlft){\tiny $\ax$};
\node at (axlinkrhtlft) [right = 2*\ilar, etic](axlinkrhtrht){\tiny $\ax$};
\node at (axlinklftlft) [below left = 2.5*\altax and \ilar, etic](premlft){};
\node at (axlinklftlft) [below right = 2.5*\altax and \ilar, etic](premrht){};
\node at (axlinklftrht) [below = 1.5*\altax, etic](lfttens){\tiny $\otimes$}; 
\node at (axlinkrhtlft) [below = 1.5*\altax, etic](rhtpar){\tiny $\parr$};
\node at (axlinkrhtrht) [below = 1.5*\altax, etic](rhttens){\tiny $\otimes$};
\node at (axlinklftrht) [below right = 4.5*\altax and \ilar, etic](cut){\tiny $\mathsf{cut}$};
\node at (lfttens) [below = 1.5*\altax, etic](sync){\tiny $\blacksquare$};
\node at (rhttens) [below = 2*\altax, etic](conc){};
\node at (sync) [etic,left=\hstlar](synlft){};
\node at (sync) [etic,right=\hstlar](synrht){};
%edges
\draw[nopol, in=90, out=180, looseness=1.5] (axlinklftlft) to (premlft);
\draw[nopol, in=135, out=0, looseness=1.5] (axlinklftlft) to (lfttens);
\draw[nopol, in=90, out=180, looseness=1.5] (axlinklftrht) to (premrht);
\draw[nopol, in=45, out=0, looseness=1.5] (axlinklftrht) to (lfttens);
\draw[nopol, in=90, out=270] (lfttens) to (sync);
\draw[nopol, in=180, out=270] (sync) to (cut);
\draw[nopol, in=135, out=180, looseness=1.5] (axlinkrhtlft) to (rhtpar);
\draw[nopol, in=135, out=0, looseness=1.5] (axlinkrhtlft) to (rhttens);
\draw[nopol, in=45, out=180, looseness=1.5] (axlinkrhtrht) to (rhtpar);
\draw[nopol, in=45, out=0, looseness=1.5] (axlinkrhtrht) to (rhttens);
\draw[nopol, in=0, out =270] (rhtpar) to (cut);
\draw[nopol, in=90, out=270] (rhttens) to (conc);
\draw[thick] (sync) to (synlft);
\draw[thick] (sync) to (synrht);
\end{tikzpicture}
$\rightarrow^*$
\begin{tikzpicture}[ocenter]
%nodes
\node at (0,0) [etic](axlinklft){\tiny $\ax$};
\node at (axlinklftlft) [right = 2*\ilar, etic](axlinkrht){\tiny $\ax$};
\node at (axlinklft) [below left = 2.5*\altax and \ilar, etic](premlft){};
\node at (axlinklft) [below = 2.5*\altax, etic](premrht){};
\node at (axlinkrht) [below left = 1.5*\altax and \ilar, etic](lftsync){\tiny $\blacksquare$};
\node at (axlinkrht) [below right = 1.5*\altax and \ilar, etic](rhtsync){\tiny $\blacksquare$};
\node at (axlinkrht) [below = 4*\altax, etic](tens){\tiny $\otimes$}; 
\node at (tens) [below = \altax, etic](conc){};
\node at (lftsync) [etic,left=\hstlar](synlft){};
\node at (rhtsync) [etic,right=\hstlar](synrht){};
%edges
\draw[nopol, in=90, out=180, looseness=1.5] (axlinklft) to (premlft);
\draw[nopol, in=90, out=0, looseness=1.5] (axlinklft) to (lftsync);
\draw[nopol, in=90, out=180, looseness=1.5] (axlinkrht) to (premrht);
\draw[nopol, in=90, out=0, looseness=1.5] (axlinkrht) to (rhtsync);
\draw[nopol, in=135, out=270] (lftsync) to (tens);
\draw[nopol, in=45, out=270] (rhtsync) to (tens);
\draw[nopol, in=90, out=270] (tens) to (conc);
\draw[thick] (lftsync) to (synlft);
\draw[thick] (rhtsync) to (synrht);
\draw[thick] (lftsync) to (rhtsync);
\end{tikzpicture}
\end{center}
\vspace{3pt}
\end{minipage}}
  \condnocaptionrule\caption{Use of $sync/\otimes$}\label{special_main}
  \end{center}
\end{figure}
Altogether, we have:
\begin{varitemize}
\item[(a)]
  the standard \MLL\ reductions: $\ax/ \cut$ and $\otimes/\parr$; 
\item[(b)]
  a set of commutations with the sync links:
  \begin{varitemize}
  \item 
    the first two rules push the sync links towards the axioms; a
    typical example of the use of the $sync/\otimes$ reduction is
    given in Figure \ref{special_main}.  they make more and more
    explicit on which atoms the synchronization is acting.
  \item 
    The $\ax/sync$ rule allows to deal with axioms which are not
    atomic. %such as the case depicted in Figure \ref{special_red}(a);
  \item 
    The $\cut/sync$ rule allow us to deal with the case in which
    $P\b$ is the lock of a box. % (Figure \ref{special_red}(b)).
  \end{varitemize}
\item[(c)] 
  Two rules which are concerned with the unit links, and
  which respectively eliminate synchronizations and boxes:
   \begin{varitemize}
   \item 
     the first rule erases a sync link whose premisses are all
     conclusions of $\onelk$ links;
   \item 
     the reduction $\onelk$/$\botlk$ \emph{opens} a box.
   \end{varitemize}
\end{varitemize}
It is important to notice that (i) there are no commutations between
sync links and that (ii) there are no commutations with a box.  
\begin{rem}[A Minimalistic Alternative] 
  Several variations on the proposed rewrite rules are indeed
  possible. Here we have made a choice of \emph{generality}, and of
  positive-negative \emph{symmetry}. If one aims at obtaining a
  minimalistic system with all the good properties (in particular,
  confluence and cut elimination), one can also choose to apply sync
  links only to positive formulas. If we assume that all axioms are
  atomic (i.e., if we make a hypothesis of $\eta$-expansion on the
  axioms), it turns out that instead of the four rules in Figure
  \ref{sync_red}(b), we only need the first one (the $sync/\otimes$
  reduction).
\end{rem}
It is now time to study the main properties of the relation $\rightarrow$.
The following has the flavour of Subject Reduction, where correctness plays
the role of typability:
\begin{lemma}[Preservation of Correctness]
  If $R$ is a correct structure and $R\rightarrow Q$, then $Q$ is
  correct, too.
\end{lemma}
By standard arguments, and by Lemma \ref{no_deadlocks}, one can prove
that no infinite sequence of reduction can be built:
\begin{prop}[Normalization]\label{termination}
  The relation $\rightarrow$ is strongly normalizing.
\end{prop}
\begin{proof}
\newcommand{\weight}{w}
Given a formula $\formone$, its \emph{weight} $\weight(\formone)$ is
the number of connectives in the formula.  Given a sync link
$\linkone$ of conclusions $\formone_1,\ldots,\formone_n$, we define
its weight $\weight(l)$ as $\sum_1^n \weight(A_i)$, and its rank
$\rho(l)$ as the number of cut and axiom links $s$ such that $s
\before l$.  Let $R$ be an \SMLL\ net. We associate to $R$ the pair
$(N_R + \weight_R, \rho_R)$, where $N_R$ is the number of nodes in
$R$, $\weight_R$ is $\sum_{l\in L} \weight(l)$ and $\rho_R$ is
$\sum_{l\in L} \rho(l)$, for $L$ the set of sync links in $R$.  We
order such pairs with the lexicographic order. Termination is
consequence of the fact that any multiplicative or unit reduction
decreases the parameter $N_R$ (possibly increasing $\rho_R$), a
commutation sync/$\otimes$ or sync/$\parr$ decreases $\weight_R$ ,
sync/$\ax$ and sync/cut strictly reduce $\rho_R$\textit{ (because
  $\before$ is a strict partial order on a finite set).}
\end{proof}
In principle, many rewrite rules can be applied to a given
structure. However, this form of nondeterminism is harmless.
\begin{prop}[Confluence]\label{confluence}
  The relation $\rightarrow$ is confluent.
\end{prop}
\begin{proof}
  The \MLL critical pairs are immediate to solve. The only critical
  pair on sync reductions is when a sync link $l$ is below a $\parr$
  and above a cut.  By inspecting this case, we see that the rewriting
  relation is \textit{locally confluent}. By using strong
  normalization, we conclude by Newman Lemma.
\end{proof}
The last two results imply that the normal form of any
\SMLL\ structure \emph{exists and is unique}. This, by itself, does
not mean that cuts can be eliminated, but is an essentially step towards
it.
%%%%%%%%%%%%%%%%%%%%%%%%%%%%
\subsection{Cut-Elimination}
%%%%%%%%%%%%%%%%%%%%%%%%%%%%
We now study cut elimination, i.e., the property that nets in normal
form contain no cuts. \SMLLcf\ turns out to have especially good
properties in this respect.

We say that a cut link at depth $0$ is \emph{ready} if neither of its
premisses is hereditary conclusion of a box. That is, above
  each premiss of the cut, there is at least a link of sort in
  \{$\axlk,\onelk,\otimes,\parr$ \} which is not inside a box. By a
straightforward case analysis, we can prove that: 
\begin{lemma}\label{readycut_lem}
 Let $R$ be an \SMLLb\ net. If there is a ready cut, then $R$ is not a normal form.
\end{lemma}
It follows immediately that \SMLLcf\ enjoys cut elimination:
\begin{theorem}[\SMLLcf\ Normal Forms]\label{cut_el}
  If $R$ is an \SMLLcf\ net in normal form, then $R$ is \emph{cut
    free}.
\end{theorem}

We now turn our attention to the whole \SMLLb. While \SMLLcf\ enjoys
cut elimination with no conditions, in the presence of boxes we
restrict our attention to the \emph{closed} case, i.e., the case in
which no $\bot$ appears in the conclusions.  We observe that the
reduction rules define a \emph{lazy} cut elimination procedure,
because there are no commutations with a box.  In the closed case,
lazy cut elimination is enough to eliminate all cuts.  The proof makes
essential use of Lemma \ref{no_deadlocks}, together with an adaptation
of Girard's analogous result for multiplicative-additive proof-nets
\cite{parallel_syntax}, which in our setting can be reformulated as
follows:
\begin{lemma}[Lazy Cut Elimination] \label{lazy_lem} 
  Let $R$ be a closed \SMLLb\ net. If $R$ is normal, then $R$ is cut-free.
\end{lemma} 
As a matter of fact, Lemma \ref{lazy_lem} is a key step towards getting
some useful information on the shape of normal forms:
\begin{theorem}[\SMLLb\ Closed Normal Forms]\label{closed_norm} 
  The normal forms of closed \SMLLb\ nets contain no cuts, no boxes,
  and no sync links.
\end{theorem}
In other words, the normal form of an \SMLLb\ net is nothing more than
a \MLL\ net! 

%%%%%%%%%%%%%%%%%%%%%
\paragraph{The Proof.}
%%%%%%%%%%%%%%%%%%%%%
We first observe the following:

\begin{lemma}[Sync Normal-Forms]\label{sync_norm_form}  
If no sync reduction  applies, then in the 0-graph of $R$ the following hold:
  \begin{varenumerate}
  \item  
    no sync link is below a $\otimes$ or a $\parr$
  \item 
    any positive premiss $P$ of a sync link is either  hereditary conclusion of a box, or has type $1$.
  \item 
    if there is a cut, only its positive premiss can be  conclusion of a  sync link.    
  \end{varenumerate}
\end{lemma}
\begin{proof}\par
\begin{varenumerate}
\item 
  Immediate:  we could apply a commutation step.
\item 
  From $P$, let us go upwards along the sync path 
  until we find a link $l$ which is not a sync link. Let us assume that $l$ is not inside a box. 
  If $l$ is of sort $\axlk$, a $sync/\axlk$ reduction would applies. Therefore $l$ can only have sort 
  $\onelk$.
\item  
  Otherwise we could apply a  $ sync/cut$ reduction. 
\end{varenumerate}
\end{proof}
\paragraph*{Proof of Lemma \ref{readycut_lem}.}
Assume that no sync reduction applies, and that there is a ready cut. If neither premiss is conclusion of a sync link, then  a multiplicative   reduction applies.
If  at least one of the premisses (say $A$) is conclusion of a sync link  $s$, then by Lemma \ref{sync_norm_form} $A$ is positive, and $A=1$. Thus, 
$A\b=\bot$  can only be conclusion of an axiom 
and an $\ax/cut$ reduction   applies.
%\condmedskip  
\paragraph*{Proof of Lemma  \ref{lazy_lem}.} 
We examine the 0-graph of $R$, and trace a path $p$ in it, as described 
below.

% net at depth 0; more precisely,  as in the condition (1) of the correctness criterion,  we treat each box of conclusion $\bot, \Gamma$ as a link of conclusions $\bot, \Gamma$,  and trace a path in it.
\begin{varenumerate}
\item
  Let us assume that  there is a cut. By hypothesis, it is not a ready cut (otherwise a reduction would apply by Lemma \ref{readycut_lem}).
  We make $p$ start from the box above the cut.  
\item
  The lock $e$ of the box has type $\bot$. We extend $p$ downwards, as follows. 
  Since no $\bot$ appears in the conclusions, $\bot$ must be subformula of a formula $C$ (possibly $\bot$ itself) which is the premiss of a cut 
  $C\b, C$. Such a cut is below $e$; we extend $p$ till this cut.  Let us examine  $C\b$. It cannot be  conclusion of a link $\axlk$ or $\onelk$, otherwise a  reduction would applies. 
  Moreover, if $C\not=\bot$, $C\b$ must be  conclusion of  a  box.
  Therefore:
  \begin{varitemize}
  \item[i] 
    either $C\b$ is  conclusion of  a  box, 
  \item[ii] 
    or $C\b = 1$ and  is a positive conclusion of a sync link $l$
  \end{varitemize}
\item 
  In both cases, we extend $p$ upwards on the edge  typed by $C\b$, 
  until  we find a node  which is either 
  \begin{varitemize}
  \item[i] a box, or 
  \item [ii] a sync link $s$ with a negative edge $N$.
  \end{varitemize}  
  Before continuing, let us prove that going upwards from $C\b$ we
  will eventually find either (i) or (ii).  We make use of Lemma
  \ref{no_deadlocks}. If $l=l_1$ is a sync link whose edges are all
  positive, and we move upwards to a link $l_2$, then either $l_2$ is
  a box, or $l_2 \before l_1$. Let us consider a maximal sequence of
  sync links $l_n\before ... \before l_2 \before l_1$. Let us assume
  that no link in this sequence has negative edges, and that $l_n$ is
  not below a box. Necessarily, all edges of $l_n$ are conclusion of
  links $\onelk$, and therefore a reduction applies, against the
  hypothesis.
\item 
  We extend $p$ depending by the case we found at step 3.  If case
  (3.i) holds, from the box we continue as in (2.); if case (3.ii)
  holds, we extend $p$ downwards on the negative conclusion $N$ of the
  sync link $s$.  Since $N$ is negative, it cannot appear in the
  conclusions, and therefore must be subformula of a formula $C$ which
  is premiss of a cut. Moreover, the cut cannot be immediately below a
  sync link (otherwise a reduction applies); thus the premiss $C\b$ of
  the cut must be conclusion of a box. We extend $p$ with the cut till
  the box, and then we continue as in (2.)
\item 
  We iterate the process: go downwards till a cut, and then
  upwards till a box, or a sync link with a negative conclusion.
  Eventually, we visit again a node or a box $l$ which we had already
  visited, and we find a cycle, by using a standard argument to study
  proof-nets: we follow the path $p$ that we have been drawing from
  $l$ to $l$ (note that $l$ does not need to be the box we started
  with, it can be any link, including a $\parr$).  We observe that $p$
  always enters a sync link on an out-edge (in fact, a positive
  conclusion) and exit on a in-edge (either a positive premiss, or a
  negative conclusion). Moreover, if $p$ reaches a $\parr$ link, $p$
  enters in it only once, from a premiss. We can conclude that $p$ is
  a switching path; by correctness, it cannot be cyclic.
\end{varenumerate}

%%%%%%%%%%%%%%%%%%%%%%%%%%%%%%%%%%%%%%%%%%%%%%%%%%%%%%%%%%%%%%%%%
\section{\SIAM: an Interactive Model with Synchronizations}
%%%%%%%%%%%%%%%%%%%%%%%%%%%%%%%%%%%%%%%%%%%%%%%%%%%%%%%%%%%%%%%%%%
The \SIAM\ is a \emph{multitoken} machine designed to run on nets of
\SMLLb.  Let us first recall the main features of the \IAM, i.e., the
standard Interaction Abstract
Machine~\cite{DanosRegnier,Mackie}. Given a net $R$, the \IAM\ pushes a
\textit{single} token around $R$. To each edge of $R$ is associated an
\emph{action} --- a \emph{transition} --- which gives instructions on
how to move the token.  A \emph{state} of the machine is a position of
the token in the net.

To define the \SIAM\ machine $M_R$ associated to a net $R$, we first
need to precisely define what an occurrence of atom and what a
position are. We can then define the states and the transitions of the
\SIAM\, and study its properties.

%%%%%%%%%%%%%%%%%%%%%%%%%%%%%%% 
\paragraph{Occurrences of Atoms.}
%%%%%%%%%%%%%%%%%%%%%%%%%%%%%%%%
\newcommand{\address}[2]{\mathtt{addr}_{#1}(#2)}
\newcommand{\lft}{\mathbf{l}} 
\newcommand{\rgt}{\mathbf{r}} 
We indicate the occurrences of atoms in a formula by their path in the
formula tree. Given an occurrence of atom $\alpha$ in a formula
$\formone$, its address $\address{\alpha}{\formone}$ in $\formone$ is
defined as a string on the alphabet $\{\lft,\rgt\}$, by induction:
\begin{varitemize}
\item 
  if $\formone=\alpha$, $\address{\alpha}{\formone}=\varepsilon$;
\item 
  if $\alpha$ is in $\formone$, then $\address{\alpha}{\formone *
    \formtwo}=\lft\cdot\address{\alpha}{\formone}$ while
  $\address{\alpha}{\formtwo *
    \formone}=\rgt\cdot\address{\alpha}{\formone}$, where $*$ is
  either $\otimes$ or $\parr$.
\end{varitemize}
With a slight abuse of notation, we sometimes identify occurrences of
atoms and their address. We indicate addresses with metavariables
like $m$.
%%%%%%%%%%%%%%%%%%%%%
\paragraph{Positions.}
%%%%%%%%%%%%%%%%%%%%%
Intuitively, each edge in a net is associated to possibly many
positions, one for each occurrence of an atom in the formula $A$
typing it. Given a net $R$, we define the set of its \emph{positions}
$\Pos(R)$ as the set including:
\begin{varitemize}
\item 
  the set of pairs $(e,m)$, where $e$ is an edge of $R$, and $m$
  is (the address of) an occurrence of atom in the formula $\formone$
  typing $e$;
\item 
  together with the set of pairs $(e,i)$, where the edge $e$ is the
  lock of a box in $R$, and $i\in \NN$. The role played by $i\in \NN$
  will be explained in Section~\ref{multiboxes}.
\end{varitemize}
We say that a position $(e,m)$ is positive or negative if it is the
case for the  occurrence of atom corresponding to $m$.  We use the
metavariables $\ss, \pp$ to indicate positions. The following subsets
of $\Pos(R)$ play a crucial role in the following:
\begin{varitemize}
\item 
  the set $\PosI(R)$ of \emph{initial positions}, that are the
  negative occurrences of atoms in the conclusions of $R$;
\item 
  the set $\PosF(R)$ of \emph{final positions}, that are the
  positive occurrences of atoms in the conclusions of $R$;
\item 
  the set $\ONE(R)$ of \emph{one positions}, which are occurrences
  of $1$ which are conclusions of $\onelk$ links;
\item 
  the set  $\BOTBOX(R)$ of pairs $(e,i)$, where 
  $e$ is the lock of a box.
\end{varitemize}
%%%%%%%%%%%%%%%%%%%%%%%%%%%%%%%%%%%%%%%%%%%
\subsection{The Machine, Formally}
%%%%%%%%%%%%%%%%%%%%%%%%%%%%%%%%%%%%%%%%%%%
It is now time to formally define the \SIAM. Given an \SMLLb\ net $R$,
we the multi-token machine $M_R$ for it consists in a set of states
and a transition relation between them. A state of $M_R$, intuitively,
tells us \emph{how many} tokens currently circulate in the net, and
which \emph{positions} they have reached from the initial state, while
keeping track (for each of them) of their \emph{origin}.
%%%%%%%%%%%%%%%%%%%
\paragraph{States.}
%%%%%%%%%%%%%%%%%%%
A \emph{state} of $M_R$ is a function $\tok:\PosI(R) \cup \pssetone
\to\Pos(R)$ where $ \pssetone \subseteq \ONE(R)$. A state is
\emph{initial} if $\tok$ is  the identity on $\PosI(R)$. We
indicate the (unique) initial state of $R$ by $\Itok_R$. A state is
\emph{final} if the image of $\tok$ is $\PosF(R) \cup B$, with $B
\subseteq \BOTBOX(R)$. Intuitively, we identify each token with its
original position $\ss$; given a state $\tok$ of $M_R$, we say that
\emph{there is} a token on the position $\pp$ if $\tok(\ss)=\pp$, for
$\ss\in\mathit{dom}(\tok)$. We use expressions such as ``a token moves'',
``crossing a link'', in the intuitive way. We will refer to the set of
all tokens as \emph{the multitoken}.
 
In describing $M_R$ we use also the following notions: an edge $e$ of
type $A$ is said to be \emph{saturated} if each position of $e$ is in
the image of $\tok$. A box is \emph{unlocked} if $e$ is the lock of
the box, and $(e,i)\in\mathit{range}(\tok)$. A link is \emph{active},
if it is either at depth $0$ in $R$, or at depth $0$ inside an
unlocked box.
%%%%%%%%%%%%%%%%%%%%%%%%%
\paragraph{Transitions.} 
%%%%%%%%%%%%%%%%%%%%%%%%%
The \emph{transition rules} are in Figure \ref{SIAM} (where
$*$ stands for either $\parr$ or $\tens$).
\begin{figure*}
  \begin{center}
    \fbox{
      \begin{minipage}{.97\textwidth}
        \begin{center}
          \input{figure_siam_rules}
        \end{center}
      \end{minipage}}
  \condnocaptionrule\caption{\SIAM\ Transition Rules}\label{SIAM}
  \end{center}
\end{figure*}
To represent the position $\pp = (e,m)$ (respectively, $\pp=(e,i)$,
$i\in\NN$) we write a bullet $\bullet$ along $e$, together with the
atom occurrence $m$ (respectively, the value $i$) next to it.  The
symbol $\down$ (respectively, $\up$) is a polarity annotation: it
indicates that the position is positive (respectively, negative). To
represent a transition $\tok \redsiam \toktwo$, we depict $\tok(\ss)$
in the left-hand-side, and $\toktwo(\ss)$ on the right-hand-side of
the arrow, for each $\ss\in\mathit{dom}(\tok)$ such that
$\toktwo(\ss)\not=\tok(\ss)$.  It is of course intended that
$\toktwo(\ss)=\tok(\ss)$ for all $\ss$ whose value is not explicitly
appearing in the picture.

Observe that the (positive or negative) polarity of any position
determines its direction: a token on a negative atom always moves
upwards, while a token on a positive atom always moves downwards. This is coherent with the way initial positions are defined. 
\begin{varenumerate}
\item 
  $\ax$, $\cut$, $\otimes$, $\parr$: the transitions are the same
  as in \MLL;
\item 
  \emph{Synchronization}. Tokens cross a sync link $l$ only when
  each in-edge of $l$ is saturated (in the pictures, if $e$ is an edge
  of type $\formone$, we write $\overline m_\formone$ for the set of
  all the occurrence of atom of $\formone$); all tokens cross the link
  simultaneously.
\item 
  The transition associated to a $\onelk$ link of conclusion $e$
  has two conditions: (i) the $\onelk$ link needs to be active, (ii)
  $\pp=(e,\varepsilon)\not\in\mathit{dom}(\tok)$.  In this case,
  $\tok$ is extended with the identity on the position $\pp$.  This is
  the only transition changing the domain of $\tok$.
\item 
  \emph{Boxes.} When the token goes through the conclusion of a box
  (graphically, the token "crosses" the border of the box), it is
  modified as if it were crossing a node:
  \begin{varitemize}
  \item[i.]  
    If there is a token on the lock $e$ of a box (necessarily
    $m=\varepsilon$), then $(e,\varepsilon)$ becomes $(e,0)\in
    \BOTBOX(R)$.  This transition plays no significant role at the
    moment, but we clarify its role in Section \ref{multiboxes}, and
    we will make essential use of it in Section~\ref{qsiam}.
  \item[ii.]  
    Tokens can enter a box only when the {box is
      unlocked}. As a consequence, if a box is not unlocked, then no
    token can be inside the box.
  \end{varitemize}
\end{varenumerate}
%%%%%%%%%%%%%%%%%%%%%%%%%%%%%%%%%%%%%%
\paragraph{Why Polarized Formulas Only?}
%%%%%%%%%%%%%%%%%%%%%%%%%%%%%%%%%%%%%%
If we synchronize the two conclusions $1,\bot$ of an atomic axiom, as
depicted in Figure \ref{fig:why_polarized_incorrect}, what we obtain
is a deadlocked structure: the sync link needs to have a token at the
same time on the in-edge of type $\bot$, and on the in-edge of type
$1$, which is not possible, since that would be the \emph{same}
token. On the other hand, this configuration is ruled out by the
correctness criterion: the structure at hand has a switching cycle.
However, if we join the two atoms with a $\parr$, and apply the sync
link on $1\parr \bot$, the situation does not change, but the
criterion does not catch the deadlock; moreover, after a step of
reduction, a cycle would appear (see Figure
\ref{fig:why_polarized_correct}).
\begin{figure}[htbp]
\centering
\fbox{
  \begin{minipage}{\figurewidth}
    \centering
    \vspace{5pt}
\subfigure[]{
\begin{tikzpicture}[ocenter]
%nodes
\node at (0,0) [etic](axlink){\tiny $\ax$};
\node at (axlink) [below left = 1.5*\altax and \ilar, etic](lftsync){\tiny $\blacksquare$};
\node at (axlink) [below right = 1.5*\altax and \ilar, etic](rhtsync){\tiny $\blacksquare$};
\node at (lftsync) [below = 1.5*\altax, etic](lftconc){$\bot$};
\node at (rhtsync) [below = 1.5*\altax, etic](rhtconc){$1$};
\node at (lftsync) [etic,left=\hstlar](synlft){};
\node at (rhtsync) [etic,right=\hstlar](synrht){};
%edges
\draw[nopol, in=90, out=180, looseness=1.5] (axlink) to (lftsync);
\draw[nopol, in=90, out=0, looseness=1.5] (axlink) to (rhtsync);
\draw[nopol, in=90, out=270] (lftsync) to (lftconc);
\draw[nopol, in=90, out=270] (rhtsync) to (rhtconc);
\draw[thick] (lftsync) to (synlft);
\draw[thick] (rhtsync) to (synrht);
\draw[thick] (lftsync) to (rhtsync);
\end{tikzpicture}\label{fig:why_polarized_incorrect}}
\qquad\qquad
\subfigure[]{
\begin{tikzpicture}[ocenter]
%nodes
\node at (0,0) [etic](axlink){\tiny $\ax$};
\node at (axlink) [below = 1.5*\altax, etic](parr){\tiny $\parr$};
\node at (parr) [below = 1.5*\altax and \ilar, etic](sync){\tiny $\blacksquare$};
\node at (sync) [below = 1.5*\altax, etic](conc){\tiny $\bot\parr 1$};
\node at (sync) [etic,left=\hstlar](synlft){};
\node at (sync) [etic,right=\hstlar](synrht){};
%edges
\draw[nopol, in=135, out=180, looseness=1.5] (axlink) to (parr);
\draw[nopol, in=45, out=0, looseness=1.5] (axlink) to (parr);
\draw[nopol, in=90, out=270] (parr) to (sync);
\draw[nopol, in=90, out=270] (sync) to (conc);
\draw[thick] (synlft) to (sync);
\draw[thick] (synrht) to (sync);
\end{tikzpicture}
$\rightarrow$
\begin{tikzpicture}[ocenter]
%nodes
\node at (0,0) [etic](axlink){\tiny $\ax$};
\node at (axlink) [below left = 1.5*\altax and \ilar, etic](lftsync){\tiny $\blacksquare$};
\node at (axlink) [below right = 1.5*\altax and \ilar, etic](rhtsync){\tiny $\blacksquare$};
\node at (axlink) [below = 3*\altax, etic](parr){\tiny $\parr$};
\node at (parr) [below = 1.5*\altax, etic](conc){\tiny $\bot\parr 1$};
\node at (lftsync) [etic,left=\hstlar](synlft){};
\node at (rhtsync) [etic,right=\hstlar](synrht){};
%edges
\draw[nopol, in=90, out=180, looseness=1.5] (axlink) to (lftsync);
\draw[nopol, in=90, out=0, looseness=1.5] (axlink) to (rhtsync);
\draw[nopol, in=135, out=270] (lftsync) to (parr);
\draw[nopol, in=45, out=270] (rhtsync) to (parr);
\draw[nopol, in=90, out=270] (parr) to (conc);
\draw[thick] (lftsync) to (synlft);
\draw[thick] (rhtsync) to (synrht);
\draw[thick] (lftsync) to (rhtsync);
\end{tikzpicture}\label{fig:why_polarized_correct}}
  \end{minipage}}
\condnocaptionrule\caption{The Need for Polarization}\label{fig:why_polarized}
\end{figure}
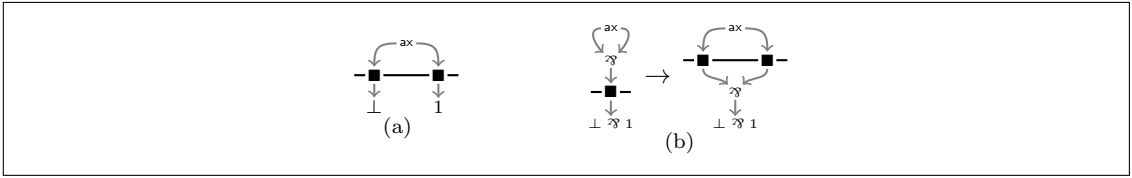

%%%%%%%%%%%%%%%%%%%%%%%%%%%%%%%%%%%%%   
\subsection{\SIAM: Properties and Soundness}
%%%%%%%%%%%%%%%%%%%%%%%%%%%%%%%%%%%%%
In this section, we study the properties of the \SIAM\, and in
particular termination, confluence, and deadlock freedom. We also show
that the \SIAM\ is (up to equivalence) a model of cut-elimination for
\SMLLb, i.e., a Soundness Theorem.  All along this section, $R$
indicates an \SMLLb\ net, $M_R$ its multitoken machine, and $\redsiam$
the induced relation.  Let us first establish some basic properties of
$\redsiam$.  %\condmedskip
\noindent
\begin{lemma}
  \begin{varenumerate}
  \item 
    Each  $\tok$  such that $\Itok_R\redsiam^* \tok$ is an \emph{injective} function.
  \item 
    There are no infinite sequences of transitions from $\Itok_R$.
  \item 
    $\redsiam$  is \emph{confluent}.
  \end{varenumerate}
\end{lemma}
%Termination is  consequence of Point 1, while confluence is consequence of the fact 
%that there are no critical pairs.
\begin{proof} 
  \begin{varenumerate}
  \item 
    About injectivity, one can observe that each position in the range
    of $\tok$ can be traced back \emph{uniquely} to a position in
    $\PosI(R)$ or to one in $\ONE(R)$ This can be proved by induction
    on the number of transitions leading from $I_R$ to $\tok$.
  \item 
    Because of finiteness, if $\redsiam$ does not terminate, it is
    because there is a loop $\tok_1 \redsiam \tok_2
    \redsiam\ldots\redsiam \tok_n \redsiam \tok_1$.  Since each
    position in the range of any $\tok_i$ can be traced back uniquely,
    one cannot access any position in the range of
    $\tok_1,\ldots,\tok_n$ from outside the loop. Therefore any
    reduction starting from the initial state cannot loop.
  \item 
    Immediate consequence of the fact that each possible pair of
    transitions commutes (diamond property), because there is no
    critical pair.
  \end{varenumerate}
 \end{proof}

A \emph{run} of the \SIAM\ machine on $R$ is a maximal sequence
$\Itok_R \redsiam \tok_1 \redsiam \tok_2 \redsiam ...\redsiam \tok_n $
of transitions from the initial state $\Itok_R$.  The lemma above
guarantees that \emph{each run of the machine $M_R$ terminates}, and
that the normal form $\tok_n$ reached from the initial state exists
and is unique.
 
A central property which we still need to prove is \textit{deadlock
  freedom}, i.e. if $\Itok_R \redsiam^* \tok$, and no reduction
applies, then $\tok$ is a final state. By confluence, and hence
unicity of the normal form, we know that if there is a run of $M_R$
which terminates on a final state $\Ftok$, then all runs of the
machine terminate on $\Ftok$. Deadlock freedom will be a consequence
of \emph{soundness}, namely that if $R$ reduces to $S$ then $M_R$ is
somehow equivalent to $M_S$. To make this precise, we define the
interpretation $[R]$ of a net $R$:
$$
[R]: (\PosI(R) \rightharpoonup \PosF(R) ) \cup \{\Box\}
$$  
We remind that $\PosI(R)$ and $\PosF(R)$ are the initial and final positions.
The formal symbol $\Box$ indicates that the machine $M_R$ deadlocks
when starting in its (unique!) initial state. Otherwise, $\Itok_R
\redsiam^* \Ftok $ a final state, and $[R]$ is defined to be the partial
function obtained as the restriction of $\Ftok$ to the elements of
$\PosI(R)$ whose image is not in $\BOTBOX(R)$.
\begin{theorem}[Soundness]\label{soundness}  
  If $R$ is a net, and $S$ is its normal form, then $[R]=[S]$.
\end{theorem}
%\condmedskip
\begin{proof}
  Let $R$ be a net of conclusions $\formone_1,\dots,\formone_n$, and
  $R\redsiam S$; to be able to compare $R$ and $S$, we adopt the
  following convention: we identify each conclusion of a net with the
  occurrence of formula $A_i$ typing it, so that there is no
  ambiguity. In particular, $R$ and $S$ have the same initial and
  final positions.  We now show that the interpretation of a net is
  preserved by all normalization steps; in particular, if $R \redsiam
  S$, then $M_R$ deadlocks if and only if $M_{S}$ deadlocks. We look
  at the reductions in Figure \ref{sync_red}.  If the multitoken never
  reaches the portion of the net which is involved in the reduction,
  then both $[R]$ and $[S]$ have the same value.  Otherwise, we
  analyze all cases, and verify that there is a deadlock in $M_R$ iff
  there is a deadlock in $M_{S}$, and otherwise, a run of the machine
  continues in a similar way, so that $[R]=[S]$. Let us just examine a
  few cases:
  \begin{varitemize}
    % \item $\otimes/\parr$ 
    % \item $\parr/sync$ lift
  \item 
    \emph{Sync Elimination}. In both the l.h.s and the r.h.s. all $\onelk$
    links are active, because they are at level 0. Thus, we are able
    to have a token on the conclusion of each $\onelk$ link, and after
    one step, on each conclusion of the sync link.
  \item 
    \emph{Box Opening}. The $\onelk$ link $l$ on the l.h.s. is active; we
    can therefore add the position $\pp$ associated with $l$ to the
    domain of the state. We now have a token on $\pp$; the token
    crosses the cut, and unlocks the box. We observe that $\pp$ does
    not belong to the domain of $[R]$. Once the box is unlocked, the
    \SIAM\ behaves exactly in the same way on the l.h.s. and the
    r.h.s.: any token reaching a position in $\Gamma$ will continue in
    exactly the same way in both sides.
  \end{varitemize}
  This concludes the proof.
  % MORE GENERAL
  % \begin{itemize}
  % \item $\otimes/\parr$ 
  % \item $\parr/sync$ lift
  % \item $sync$ elimination.\textbf{[ ***this is an interesting case***]}   In both l.h.s and r.h.s. either all $\onelk$ links are active, or none of them, because they are all at level 0, or inside the same box  %(*** all at level 0, if we add the condition in bold below ***)
  %   
  %   If  in both l.h.s and r.h.s. all $\onelk$ links are active, then we  are able to have a token on the conclusion  of each $\onelk$ link (and in one step,  on each conclusion of the sync link). If no,   no token can be present in either side.
  %   
  % \item box open.   If the $\onelk$ link on the l.h.s. is not active, it means that we are inside a box which is unlocked; hence there is also no token circulating in the depicted portion of the net in either the l.h.s. or the  r.h.s.
  %   If the $\onelk$ link is active, we add the position $\pp$ on its conclusion to the domain of the state. The token crosses  the cut, and unlocks the box;  $\pp$ does not belong to the domain of $f_R$. Once the box is unlocked,the l.h.s. and the r.h.s. behave exactly in the same way: any token reaching a position in $\Gamma$ will continue in exactly the same way in both sides.
  % \end{itemize}
\end{proof}

A consequence of Theorem~\ref{soundness} is that $M_R$ deadlocks iff
$M_{S}$ deadlocks.  Soundness, thus, allows us to study deadlock
freedom of the \SIAM\ by studying deadlock freedom of the \SIAM\ when
running on a cut-free net.
%\condmedskip
\begin{theorem}[Termination and Deadlock Freedom]\label{final_th}
  Let $R$ be a net, and $M_R$ its multitoken machine. All runs of
  $M_R$ terminate on a final state $\Ftok_R$. Moreover, such a final
  state is unique.
\end{theorem}
%\condmedskip
This result holds for both \SMLLcf\ and \SMLLb\ nets. However, the
argument is quite different, because the results which we have on cut
elimination are themselves different:
\begin{varenumerate}
\item 
  In the case of \SMLLcf, we can exploit the fact that the normal
  form of a net is cut-free (Theorem~\ref{cut_el}). The difficulty
  comes from the fact that in the normal form we do have sync links.
  However, if a net is cut free, all sync links are hereditary
  conclusions of $\axlk$ links, and the correctness criterion (via
  Lemma~\ref{no_deadlocks}) allows us to establish that the tokens
  cannot get stuck before reaching a final state.  Together with
  Soundness, this allows us to conclude.
\item 
  In the case of \SMLLb, on the one hand the proof is simplified
  by the fact that all sync links can be eliminated, but on the other
  hand, we have a result of cut elimination which is limited to the
  closed case. We modularize the proof into two steps:
  \begin{varitemize}
  \item[i.]  
    if $R$ is an \SMLLb\ net with no $\bot$ in the
    conclusion, Theorem \ref{final_th} is an immediate consequence of
    Soundness and Theorem \ref{closed_norm};
  \item[ii.]  
    we then prove the result for an arbitrary \SMLLb\ net
    $R$ by providing a construction that embeds $R$ into a larger net
    $\widehat R$ (its closure) which has no $\bot$ in its conclusion
    (therefore Point 1. applies), and by showing that $M_{\widehat R}$
    deadlocks iff $M_R$ deadlocks.
  \end{varitemize}
\end{varenumerate}
One may wonder what happens to the sync links in case (2.i).
Normalization pushes the sync's upwards, and since we are in the
closed case, eventually all sync links are hereditary conclusions of
$\onelk$ links. The elimination of the sync links 
(Theorem~\ref{closed_norm}) can thus be seen as an step towards
deadlock freedom: Lemma~\ref{no_deadlocks} guarantees that no situation
like the one in Figure~\ref{fig_nodeadlock}(b) can arise (thus, there is
always a top-most sync link which can be eliminated).
%%%%%%%%%%%%%%%%%%%%%
\paragraph{On The Proof of Theorem~\ref{final_th}.}
%%%%%%%%%%%%%%%%%%%%%
Some more details can now be given about the proof of
Theorem~\ref{final_th}.
\begin{lemma}[\SMLLcf] 
  Let $R$ be a \SMLLcf\ net. All runs of the machine $M_R$ terminate
  on a final state.
\end{lemma}
%\condmedskip
\begin{proof} 
  We examine the normal form of $R$, and establish the result by using
  Theorem \ref{soundness} and Lemma \ref{no_deadlocks}.
%\TODO{
%Since  only sync links can cause a deadlock,
%if there are no sync links in the net, no deadlock is possible. Hence we  prove that the multitoken can transit  across all sync links of $R'$,  by induction on the number $n$ of the sync links which have not yet been used in a transition.
%To simplify the argument, we transform $R$ into a net  $R'$, which is the one obtained  by $Nsync/ax$ reductions (we just move the sync links from the negative to the positive conclusion of the $\ax$ links).
% It is immediate that $f_R=f_{R'}$.
%% All sync links act on positive formulas, which are hereditary conclusion of either $1$ or $\ax$.  
%% 
%\begin{enumerate}
%\item 
%We observe that from the initial state of $M_{R'}$, we can reach a state  where there is a token on each conclusion of a $\onelk$ link (immediate), and the positive  conclusion  of each   $\ax$ link are  saturated (because  all tokens starting from $\PosI$  can go up till the axioms, without crossing any sync link).
%\item We choose a sync link $l$ which is minimal w.r.t. the before relation (before and above here coincide). By \ref{no_deadlocks}, such a sync link exists.  All the IN edges of $l$ are positive, and conclusion of either an $\ax$ or a $\onelk$, hence they can be saturated, and the transition applies. 
%\item Eventually, the tokens have crossed all sync links.
%\end{enumerate}
%We conclude by confluence.
%}
\end{proof}
We can prove the same for closed \SMLLb\ nets, as an easy consequence
of Corollary \ref{closed_norm}:
%\condmedskip
\begin{lemma}[\SMLLb: Closed Case]\label{closed_SIAM} 
  Let $R$ be a closed \SMLLb\ net. All runs of $M_R$ terminate on a
  final state.
\end{lemma}
%\condmedskip
But we can go even further, and prove the same result for arbitrary
\SMLLb\ nets by following another strategy.  Given an \MLL\ formula
$A$, we associate to it a net $R^{\circ}(A)$ which has conclusions
$\Gamma, A$, and the property that no formula in $\Gamma$ contains
occurrences of $\bot$. We call $A$ the main conclusion of
$R^{\circ}(A)$.  We can proceed as follows, by induction on the
structure of $A$:
\begin{varitemize}
\item 
  $R^{\circ}(1)$ is the $\onelk$ link;
  % \item $R^{\circ}(\bot)$ is the axiom link $\vdash 1, \bot$;
\item 
  $R^{\circ}(\at)$, where $\at$ is either $\bot$ or $X$ or $X\b$,
  is the axiom link of conclusions $\at\b , \at$.
\item 
  $R^{\circ}(A_1\otimes A_2)$ is the net of conclusions
  $\Gamma_1,\Gamma_2, A_1\otimes A_2 $ which is obtained from
  $R^{\circ}(A_1)$ and $R^{\circ}(A_1)$ by connecting the main
  conclusions with a $\otimes$ link.
\item 
  $R^{\circ}(A_1\parr A_2)$ is the net of conclusions
  $\Gamma_1,\Gamma_2, A_1\parr A_2 $ which is obtained from
  $R^{\circ}(A_1)$ and $R^{\circ}(A_1)$ by connecting the main
  conclusions with a $\parr$ link.
\end{varitemize}
Given a \SMLLb\ net $R$, its \emph{closure} $\widehat R$ is the net resulting
by cutting each conclusion $A$ of $R$ which contains $\bot$ with the
main conclusion of the net $R^{\circ}(A\b)$. Since each $R^{\circ}(A\b)$ is a
multiplicative net, it is easy to see that if there is a deadlock in
$\widehat R$, it can only be inside $R$.
\begin{lemma}\label{smllbclosure} 
  Let $R$ be a \SMLLb\ net and $\widehat R$ its closure.  $M_R$ is
  deadlock free if and only if $M_{\widehat R}$ is deadlock free.
\end{lemma}
As a consequence:
\begin{lemma}[\SMLLb: General Case]\label{df_smllb} 
  Let $R$ be an \SMLLb\ net. All runs of $M_R$ terminate on a  final state.
\end{lemma}
\begin{proof}
  Immediate corollary of Lemma~\ref{smllbclosure} and Lemma \ref{closed_SIAM}.
\end{proof}
%%%%%%%%%%%%%%%%%%%%%%%%%%%%%%%%%%%%%%%%%%%
\subsection{Multi-Boxes}\label{multiboxes}
%%%%%%%%%%%%%%%%%%%%%%%%%%%%%%%%%%%%%%%%%%%
We have everything in place to model also probabilistic or
non-deterministic choice, with only a small modification of \SMLLb,
where the content of a box can be not a \emph{single} net, but a
\emph{sequence} of $n>1$ nets. A multi-box of conclusion $\bot,
\Gamma$ may contain several nets $R_1,\ldots,R_{n}$, all of the same
conclusion $\Gamma$ (see Figure~\ref{fig:multiboxes}). 
\begin{figure}[htbp]
\centering
\fbox{
  \begin{minipage}{\figurewidth}
    \vspace{5pt}
    \begin{center}
      \begin{tikzpicture}[ocenter]
% bang
\node at (0,0) [etic](pospal){};
\node at (pospal.center)[right=2.7*\stlar, etic](posauxportu){};
\abox{pospal}{exbox}{\stboxlw}{\stboxrwaux+40pt}{\stboxh+7pt}
\node at (pospal.center) [etic](possym){$\bot$};
% auxiliary ports
\node at (pospal.center)[right=1.5*\stlar, etic](posauxportl){\tiny $\Gamma$};
\node at (pospal.center)[right=4*\stlar, etic](posauxportr){\tiny $\Gamma$};
\node at (pospal.center)[below right=0.5*\stalt and 2.7*\stlar, etic](posauxportd){\tiny $\Gamma$};
% net g
\node at (pospal.center)[above right=0.8*\stalt and 1.5*\stlar, net](netl){$R_1$};
\node at (pospal.center)[above right=0.8*\stalt and 2.8*\stlar, etic](netc){$\cdots$};
\node at (pospal.center)[above right=0.8*\stalt and 4*\stlar, net](netr){$R_n$};
\draw[nopolgen, in=90, out=-90](netl)to(posauxportl);
\draw[nopolgen, in=90, out=-90](netr)to(posauxportr);
\draw[nopolgen, in=90, out=-90](posauxportu)to(posauxportd);
% bot node
\node at (possym.center) [above=1.4*\altax, etic](botlink){\tiny $\mathsf{bot}$};
\draw[nopol, in=90, out=-90] (botlink) to (possym);
\node at (possym.center) [below=0.5*\stalt](possymd){};
\draw[nopol, in=90, out=-90] (possym) to (possymd);
\end{tikzpicture}
    \end{center}
  \end{minipage}}
\condnocaptionrule\caption{A Multi-box}\label{fig:multiboxes}
\end{figure}
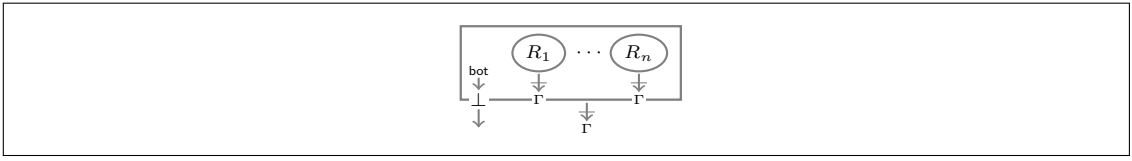
In sequent calculus, this would correspond to the rule
$$
\infer{\vdash \bot, \Gamma}
{{\vdash \Gamma} & \cdots &  {\vdash \Gamma}}
$$
The normalization relation becomes nondeterministic, as it chooses one
of the nets inside the box when opening it. All the interesting
properties hold, with very similar proofs.

The \SIAM\ needs to be adapted slightly, but in fact was already set
for this (see point 4 in the discussion on transitions).  If a box
contains $n$ nets, and a token reaches the lock of the box, the
machine choose a value $i \in \{0,\ldots,n-1\}$ nondeterministically;
the integer $i$ now determines to which of the $R_i$ internal nets the
incoming tokens will move (only $R_i$ is active). As a consequence,
for any box, at most one of the internal nets is populated and
traversed by tokens. If we call \emph{slice} a choice of a single
subnet for each box, one can check that a run of the machine always
happens in a single slice, and has therefore similar properties to the
ones we studied above. We will put this at work in the following
section.
%%%%%%%%%%%%%%%%%%%%%%%%%%%%%%%%%%%%%%%%%%%%%%%%%%%%%%%%%%%%%
\section{Synchronization, Proof-nets, and Quantum Computation}\label{quantum_section}
%%%%%%%%%%%%%%%%%%%%%%%%%%%%%%%%%%%%%%%%%%%%%%%%%%%%%%%%%%%%%
In this section, the whole development of the last two sections will be applied to higher-order quantum
computation in the form of a linear quantum $\lambda$-calculus akin to that recently introduced
by Selinger and Valiron~\cite{SelingerValiron}.
%%%%%%%%%%%%%%%%%%%%%%%%%%%%%%%%%%%%%%%%%%%%%%
\subsection{A Linear Quantum $\lambda$-calculus}
%%%%%%%%%%%%%%%%%%%%%%%%%%%%%%%%%%%%%%%%%%%%%%
\newcommand{\pof}{\triangleright}
\newcommand{\tdone}{\pi}
\newcommand{\tdtwo}{\psi}
\newcommand{\tdthree}{\xi}
\newcommand{\qvars}{\mathcal{QV}}
\newcommand{\unops}{\mathcal{UO}}
\newcommand{\fsvone}{V}
\newcommand{\fsvtwo}{W}
\newcommand{\ccontone}{\Theta}
\newcommand{\uopone}{\mathsf{U}}
\newcommand{\uoptwo}{\mathsf{V}}
\newcommand{\tensp}[2]{#1^{\otimes #2}}
\newcommand{\fsone}{X}
\newcommand{\qrone}{Q}
\newcommand{\qrtwo}{P}
\newcommand{\sbsts}[1]{\mathcal{SUB}(#1)}
\newcommand{\hilb}[1]{\mathbb{H}(#1)}
\newcommand{\cls}{\mathcal{C}}
\newcommand{\clsp}[1]{\mathcal{C}_{#1}}
\newcommand{\RR}{\mathbb{R}}
\newcommand{\qred}{\mapsto}
\newcommand{\qredmult}{\Rightarrow}
\newcommand{\qcone}{C}
\newcommand{\qctwo}{D}
\newcommand{\distrone}{\mathscr{D}}
\newcommand{\distrtwo}{\mathscr{E}}
\newcommand{\distrs}[1]{\mathbb{D}(#1)}
\newcommand{\probone}{p}
\newcommand{\probtwo}{q}
\newcommand{\tuple}[1]{\langle #1\rangle}
\newcommand{\pofobs}[3]{\mathcal{R}^{#2}_{#3}(#1)}
\newcommand{\proj}[3]{\mathcal{J}^{#2}_{#3}(#1)}
\newcommand{\valone}{a}
\newcommand{\valtwo}{b}
\newcommand{\valthree}{c}
\newcommand{\bitone}{b}
\newcommand{\econe}{E}
\newcommand{\ectwo}{F}
\newcommand{\ecthree}{G}
\newcommand{\emctx}{[\cdot]}
\newcommand{\apec}[2]{#1[#2]}
\newcommand{\apecd}[1]{\econe[#1]}
\newcommand{\suprt}[1]{\mathsf{S}(#1)}
\newcommand{\normform}{\mathit{NF}}
\newcommand{\setone}{X}
\newcommand{\settwo}{Y}
We assume given a finite set $\unops$ of symbols, each denoting a unitary operator.
This is ranged over by  metavariables like $\uopone$ and $\uoptwo$. When we want to 
insist on $\uopone$ having arity $n$, we write it as $\uopone_n$.
\emph{Terms}  are defined as follows:
\begin{align*}
\termone,\termtwo,\termthree&\bnf
    \varone\midd\qvarone\midd\termone\termtwo\midd\lam x.t\midd\pair{t}{u}\midd\llet{x}{y}{t}{u}\\
    &\midd\new\midd\ttrue\midd\tfalse\midd\meas\midd\uopone\midd\ifthen{t}{u}{v}.
\end{align*}
where $\qvarone$ ranges in a denumerable set $\qvars$ of quantum variables.
\emph{Types} are defined as follows:
$$
\typone,\typtwo\bnf\bool\midd\qbool\midd\typone\lin\typtwo\midd\typone\otimes\typtwo.
$$
For every natural number $n$ and for every type $\typone$, the expression
$\tensp{\typone}{n}$ stands for the type
$$
\underbrace{\typone\tens\cdots\tens\typone}_{n\mbox{ times}}.
$$
\emph{Typing judgments} are in the form
$\tyg{\ccontone,\fsvone}{\termone}{\typone}$ where $\termone$ is a term, $\typone$ is a type,
$\ccontone$ is an environment mapping $\lambda$-variables to types, and $\fsvone\subseteq\qvars$.
Expressions like $\ccontone,\fsvone$ are called \emph{contexts} and are denoted by metavariables
like $\contone$ or $\conttwo$. \emph{Typing rules} are as in Figure~\ref{fig:typingrules}.
\begin{figure}
\begin{center}
\fbox{
\footnotesize
\begin{minipage}{\figurewidth}
$$
\AxiomC{}
\UnaryInfC{$\tyg{\varone:\typone}{\varone}{\typone}$}
\DisplayProof
\qquad
\AxiomC{}
\UnaryInfC{$\tyg{\qvarone}{\qvarone}{\qbool}$}
\DisplayProof
\qquad
\AxiomC{$\tyg{\contone,\varone:\typone}{\termone}{\typtwo}$}
\UnaryInfC{$\tyg{\contone}{\lam\varone.\termone}{\typone\lin\typtwo}$}
\DisplayProof
$$
$$
\AxiomC{$\tyg{\contone}{\termone}{\typone\lin\typtwo}$}
\AxiomC{$\tyg{\conttwo}{\termtwo}{\typone}$}
\BinaryInfC{$\tyg{\contone,\conttwo}{\termone\termtwo}{\typtwo}$}
\DisplayProof
\quad
\AxiomC{$\tyg{\contone}{\termone}{\typone}$}
\AxiomC{$\tyg{\conttwo}{\termtwo}{\typtwo}$}
\BinaryInfC{$\tyg{\contone,\conttwo}{\pair{\termone}{\termtwo}}{\typone\tens\typtwo}$}
\DisplayProof
$$
$$
\AxiomC{$\tyg{\contone}{\termone}{\typone\tens\typtwo}$}
\AxiomC{$\tyg{\conttwo,\varone:\typone,\vartwo:\typtwo}{\termtwo}{\typthree}$}
\BinaryInfC{$\tyg{\contone,\conttwo}{\llet{\varone}{\vartwo}{\termone}{\termtwo}}{\typthree}$}
\DisplayProof
$$
$$
\AxiomC{}
\UnaryInfC{$\tyg{\cdot}{\ttrue}{\bool}$}
\DisplayProof
\quad
\AxiomC{}
\UnaryInfC{$\tyg{\cdot}{\tfalse}{\bool}$}
\DisplayProof
\quad
\AxiomC{$\tyg{\contone}{\termone}{\bool}$}
\AxiomC{$\tyg{}{\termtwo}{\typone}$}
\AxiomC{$\tyg{}{\termthree}{\typone}$}
\TrinaryInfC{$\tyg{\contone}{\ifthen{\termone}{\termtwo}{\termthree}}{\typone}$}
\DisplayProof
$$
$$
\AxiomC{}
\UnaryInfC{$\tyg{\cdot}{\new}{\qbool}$}
\DisplayProof
\quad
\AxiomC{}
\UnaryInfC{$\tyg{\cdot}{\uopone_n}{\qbool^{\otimes n}\lin \qbool^{\otimes n} }$}
\DisplayProof
\quad
\AxiomC{}
\UnaryInfC{$\tyg{\cdot}{\meas}{\qbool\lin\bool}$}
\DisplayProof
$$
\vspace{0pt}
\end{minipage}}
\condnocaptionrule\caption{Typing Rules}\label{fig:typingrules}
\end{center}
\end{figure}
Observe how terms in the form
$\ifthen{\termone}{\termtwo}{\termthree}$ are typed in a slightly
non-standard way, in that $\termtwo$ and $\termthree$ are required to
be \emph{closed}. In presence of higher-order types and nameless
functions, however, this does not cause any significant threat to
expressivity. If a derivation $\tdone$ has conclusion
$\tyg{\contone}{\termone}{\typone}$, then we write
$\tdone\pof\tyg{\contone}{\termone}{\typone}$. \emph{Values} are terms
generated by the following grammar:
$$
\valone,\valtwo\bnf\qvarone\midd\abstr{\varone}{\termone}\midd\pair{\valone}{\valtwo}\midd
\new\midd\ttrue\midd\tfalse\midd\meas\midd\uopone.
$$ 
\emph{Evaluation contexts} are expressions built as follows:
\begin{align*}
\econe,\ectwo,\ecthree&\bnf
\emctx\midd\econe\termone\midd\valone\ectwo\midd\pair{\econe}{\termone}\midd\pair{\valone}{\ectwo}
    \midd\llet{x}{y}{\econe}{\termone}\midd\ifthen{\econe}{\termone}{\termtwo}.
\end{align*}
As usual, we indicate as $\apec{\econe}{\termone}$ the term obtained
by filling the only occurrence of $\emctx$ in $\econe$ with
$\termone$.
%%%%%%%%%%%%%%%%%%%%%%%%%%%%%%%%%%%%%%%%%%%%%%%%%%%%%%%%%%%
\subsection{Quantum Closures and Operational Semantics}\label{sect:qcos}
%%%%%%%%%%%%%%%%%%%%%%%%%%%%%%%%%%%%%%%%%%%%%%%%%%%%%%%%%%%
Given any finite set $\fsone$, $\hilb{\fsone}$ is the
finite-dimensional Hilbert space generated by $\fsone$, i.e.  the
Hilbert Space having $\fsone$ as basis. Let $\fsone$ be any finite
set. By $\sbsts{\fsone}$ we note the (finite) set of all functions
mapping elements of $\fsone$ to bits in $\{0,1\}$. Please observe that
$|\sbsts{\fsone}|=2^{|\fsone|}$. Given an element $\qrone$ of
$\hilb{\sbsts{\fsone}}$, an element $\qvarone\in\fsone$ and a bit
$\bitone\in\{0,1\}$, we indicate with
$\proj{\qrone}{\qvarone}{\bitone}$ the projection of $\qrone$ into the
subspace of $\hilb{\sbsts{\fsone}}$ in which $\qvarone$ is assigned to
$\bitone$. Similarly, $\pofobs{\qrone}{\qvarone}{\bitone}$ is the
probability of observing $\bitone$ when measuring the value of
$\qvarone$ in $\qrone$.

A \emph{quantum closure} of type $\typone$ is a pair
$[\qrone,\termone]$, where $\qrone$ is a normalized vector in $\hilb{\sbsts{\fsvone}}$
and $\tyg{\fsvone}{\termone}{\typone}$. In this case we often
write $[\qrone,\termone]:\typone$. Quantum closures are taken modulo a form of $\alpha$-equivalence
in which one is allowed to change the name of a quantum variable, modifying the underlying quantum register
accordingly. Quantum closures are denoted with 
metavariables like $\qcone$ and $\qctwo$. $\clsp{\typone}$ is the set of all quantum closures
of type $\typone$, while $\cls$ is the set of all quantum closures.

\begin{sloppypar}
We need to work with distributions on quantum closures, namely functions 
in the form $\distrone:\cls\rightarrow\RR_{[0,1]}$ satisfying $\sum_{\qcone\in\cls}\distrone(\qcone)\leq 1$.
The subset of $\cls$ of those quantum closures $\qcone$ for which $\distrone(\qcone)>0$ is said
to be the \emph{support} of $\distrone$ and is denoted as $\suprt{\distrone}$. In the following,
we will only be concerned with distributions having (at most) denumerable support.
Distributions having a finite support are indicated with expressions
like $\{\qcone_1^{\probone_1},\ldots,\qcone_n^{\probone_n}\}$, with the obvious meaning.
The set of distributions over a set $\setone$ is denoted as $\distrs{\setone}$.
\end{sloppypar}

Quantum closures can be given a semantics in two steps:
first, one gives a binary relation $\qred$ between quantum closures and distributions capturing
one-step reduction, then one generalizes the relation above to another relation $\qredmult$ 
capturing multi-step reduction. Rules for $\qred$ are in Figure~\ref{fig:opsem}.
\begin{figure*}
\begin{center}
\fbox{
\begin{minipage}{.97\textwidth}
$$
\AxiomC{}
\UnaryInfC{$[\qrone,\apecd{(\abstr{\varone}{\termone})\termtwo}]\qred\{[\qrone,\apecd{\subst{\termone}{\varone}{\termtwo}}]^1\}$}
\DisplayProof
\qquad
\AxiomC{}
\UnaryInfC{$[\qrone,\apecd{\llet{\varone}{\vartwo}{\pair{\termone}{\termtwo}}{\termthree}}]
  \qred\{[\qrone,\apecd{\subst{\termthree}{\varone,\vartwo}{\termone,\termtwo}}]^1\}$}
\DisplayProof
$$
$$
\AxiomC{}
\UnaryInfC{$[\qrone,\apecd{\ifthen{\ttrue}{\termone}{\termtwo}}]\qred\{[\qrone,\apecd{\termone}]^1\}$}
\DisplayProof
\qquad
\AxiomC{}
\UnaryInfC{$[\qrone,\apecd{\ifthen{\tfalse}{\termone}{\termtwo}}]\qred\{[\qrone,\apecd{\termtwo}]^1\}$}
\DisplayProof
$$
$$
\AxiomC{$\qvarone$ is fresh}
\UnaryInfC{$[\qrone,\apecd{\new}]\qred[\qrone\tens\ket{\qvarone\leftarrow 0},\apecd{\qvarone}]$}
\DisplayProof
\qquad
\AxiomC{}
\UnaryInfC{$[\qrone,\apecd{\uopone_n\tuple{\qvarone_1,\ldots,\qvarone_n}}]\qred
  \{[\uopone_n^{\varone_1,\ldots,\varone_n}(\qrone),\apecd{\tuple{\qvarone_1,\ldots,\qvarone_n}}]^1\}$}
\DisplayProof
$$
$$
\AxiomC{}
\UnaryInfC{$[\qrone,\apecd{\meas\;\qvarone}]\qred
  \{
    [\proj{\qrone}{\qvarone}{0},\apecd{\ttrue}]^{\pofobs{\qrone}{\qvarone}{0}},
    [\proj{\qrone}{\qvarone}{1},\apecd{\tfalse}]^{\pofobs{\qrone}{\qvarone}{1}}
  \}$}
\DisplayProof
$$
\vspace{0pt}
\end{minipage}}
\condnocaptionrule\caption{Operational Semantics}\label{fig:opsem}
\end{center}
\end{figure*}
Reduction along $\qred$ preserves types:
%\condmedskip
\begin{lemma}[Subject Reduction]
  If $\qcone:\typone$ and $\qcone\qred\distrone$, then $\qctwo:\typone$ for
  every $\qctwo\in\suprt{\distrone}$.
\end{lemma}
%\condmedskip
Subject reduction can be proved with the usual strategy, namely through an appropriate Substitution Lemma:
\begin{lemma}[Substitution Lemma for Terms]\label{lem:substLemTerms}
  If $\tyg{\contone, \varone:\typone}{\termone}{\typtwo}$ and
  $\tyg{\conttwo}{\termtwo}{\typone}$,
  then $\tyg{\contone, \conttwo}{\subst{\termone}{\varone}{\termtwo}}{\typtwo}$.
\end{lemma}
\begin{proof}
  This is the usual induction on the structure of a proof of
  $\tyg{\contone, \varone:\typone}{\termone}{\typtwo}$.
\end{proof}
\begin{deff}
  We write $\tyg{\contone}{\econe}{\typone \ecimp \typtwo}$
  when an evaluation context $\econe$ satisfies the following:
  if $\tyg{\conttwo}{\termone}{\typone}$ then
  $\tyg{\contone,\conttwo}{\apec{\econe}{\termone}}{\typtwo}$.
\end{deff}
Such judgments for evaluation contexts can be derived by the inference rules in Figure~\ref{fig:evalCtxtRules};
indeed the rules are admissible.
\begin{figure}[htbp]
  \centering
  \fbox{
    \footnotesize
    \begin{minipage}{\figurewidth}
    $$
    \AxiomC{}
    \UnaryInfC{$\tyg{\cdot}{\emctx}{\typone \ecimp \typone}$}
    \DisplayProof
    \qquad
    \AxiomC{$\tyg{\contone}{\econe}{\typone \ecimp (\typthree \imp \typtwo)}$}
        \AxiomC{$\tyg{\conttwo}{\termone}{\typthree}$}
    \BinaryInfC{$\tyg{\contone,\conttwo}{\econe \termone}{\typone \ecimp \typtwo}$}
    \DisplayProof
    $$
    $$
    \AxiomC{$\tyg{\conttwo}{\valone}{(\typthree \imp \typtwo)}$}
        \AxiomC{$\tyg{\contone}{\ectwo}{\typone \ecimp \typthree}$}
    \BinaryInfC{$\tyg{\contone,\conttwo}{\valone \ectwo}{\typone \ecimp \typtwo}$}
    \DisplayProof
    $$
    $$
    \AxiomC{$\tyg{\contone}{\econe}{\typone \ecimp \typtwo}$}
        \AxiomC{$\tyg{\conttwo}{\termone}{\typthree}$}
    \BinaryInfC{$\tyg{\contone,\conttwo}
      {\pair{\econe}{\termone}}
      {\typone \ecimp \typtwo \tens \typthree}$}
    \DisplayProof
    \qquad
    \AxiomC{$\tyg{\contone}{\valone}{\typtwo}$}
        \AxiomC{$\tyg{\conttwo}{\ectwo}{\typone \ecimp \typthree}$}
    \BinaryInfC{$\tyg{\contone,\conttwo}
      {\pair{\valone}{\ectwo}}
      {\typone \ecimp \typtwo \tens \typthree}$}
    \DisplayProof
    $$
    $$
    \AxiomC{$\tyg{\contone}{\econe}{\typone \ecimp \typtwo \tens \typthree}$}
        \AxiomC{$\tyg{\conttwo,\varone:\typtwo,\vartwo:\typthree}{\termone}{\typefour}$}
    \BinaryInfC{$\tyg{\contone,\conttwo}
      {\llet{\varone}{\vartwo}{\econe}{\termone}}
      {\typone \ecimp \typefour}$}
    \DisplayProof
    $$
    $$
    \AxiomC{$\tyg{\contone}{\econe}{\typone \ecimp \bool}$}
        \AxiomC{$\tyg{\cdot}{\termone}{\typthree}$}
            \AxiomC{$\tyg{\cdot}{\termtwo}{\typthree}$}
    \TrinaryInfC{$\tyg{\contone}
      {\ifthen{\econe}{\termone}{\termtwo}}
      {\typone \ecimp \typthree}$}
    \DisplayProof
    $$
    \end{minipage}
  }
\caption{Evaluation Contexts --- Typing Rules}\label{fig:evalCtxtRules}
\end{figure}
\begin{lemma}\label{lem:evalCtxtInversion}
  If $\tyg{\contone}{\apec{\econe}{\termone}}{\typone}$ then
  $\tyg{\conttwo}{\termone}{\typtwo}$ and $\tyg{\ctxthree}{\econe}{\typtwo \ecimp \typone}$
  for some type $\typtwo$ and contexts $\conttwo$, $\ctxthree$
  where $\contone = \conttwo, \ctxthree$.
  \hfill $\Box$
\end{lemma}
\begin{proof}(of Subject Reduction)
  The proof is done by a case analysis on the form of the reduction step $\qcone\qred\distrone$:
  \begin{varitemize}
  \item 
    Suppose it is
    $[\qrone,\apecd{(\abstr{\varone}{\termone})\termtwo}]
    \qred\{[\qrone,\apecd{\subst{\termone}{\varone}{\termtwo}}]^1\}$
    and that
    $\tyg{\contone}{\apecd{(\abstr{\varone}{\termone})\termtwo}}{\typeone}$.
    Then by Lemma~\ref{lem:evalCtxtInversion}
    $\tyg{\conttwo}{(\abstr{\varone}{\termone})\termtwo}{\typetwo}$
    and $\tyg{\ctxthree}{\econe}{\typetwo \ecimp \typeone}$ for some
    type $\typetwo$ and contexts $\conttwo$, $\ctxthree$ where
    $\contone = \conttwo, \ctxthree$.  Since typing rules are
    syntax-directed (i.e.\ any judgment uniquely determines its
    derivation if derivable), the derivation of
    $\tyg{\conttwo}{(\abstr{\varone}{\termone})\termtwo}{\typetwo}$ is
    \begin{prooftree}
      \AxiomC{$\pi_1$}
      \noLine
      \UnaryInfC{$\tyg{\conttwo_1, \varone:\typetwo'}{\termone}{\typetwo}$}
      \UnaryInfC{$\tyg{\conttwo_1}{(\abstr{\varone}{\termone})}{\typetwo' \imp \typetwo}$}
          \AxiomC{$\pi_2$}
          \noLine
          \UnaryInfC{$\tyg{\conttwo_2}{\termtwo}{\typetwo'}$}
      \BinaryInfC{$\tyg{\conttwo}{(\abstr{\varone}{\termone})\termtwo}{\typetwo}$}
    \end{prooftree}
    where $\conttwo = \conttwo_1, \conttwo_2$.
    Thus by Lemma~\ref{lem:substLemTerms}
    $\tyg{\conttwo}{\subst{\termone}{\varone}{\termtwo}}{\typtwo}$
    and
    $\tyg{\contone}{\apecd{\subst{\termone}{\varone}{\termtwo}}}{\typone}$
    by definition of $\typetwo \ecimp \typeone$.
  \item 
    Suppose that
    $[\qrone,\apecd{\llet{\varone}{\vartwo}{\pair{\termone}{\termtwo}}{\termthree}}]
    \qred\{[\qrone,\apecd{\subst{\termthree}{\varone,\vartwo}{\termone,\termtwo}}]^1\}$.
    The proof is similar to the above case.
  \item 
    Suppose it is
    $[\qrone,\apecd{\ifthen{\ttrue}{\termone}{\termtwo}}]
    \qred\{[\qrone,\apecd{\termone}]^1\}$, and that
    $\tyg{\contone}{\apecd{\ifthen{\ttrue}{\termone}{\termtwo}}}{\typeone}$.
    Then by Lemma~\ref{lem:evalCtxtInversion}
    $\tyg{\cdot}{\ifthen{\ttrue}{\termone}{\termtwo}}{\typetwo}$ (the
    context is uniquely determined to be $(\cdot)$ by typing rules in
    this case) and $\tyg{\contone}{\econe}{\typetwo \ecimp \typeone}$
    for some $\typetwo$.  Since typing rules are syntax-directed, it
    can be immediately checked that the derivation of
    $\tyg{\cdot}{\ifthen{\ttrue}{\termone}{\termtwo}}{\typetwo}$
    contains a derivation of $\tyg{\cdot}{\termone}{\typetwo}$.  Thus
    $\tyg{\contone}{\apecd{\termone}}{\typeone}$ by definition of
    $\typetwo \ecimp \typeone$.
  \item 
    The other cases are proved similarly:
    $\tyg{\ctxthree}{\econe}{\typetwo \ecimp \typeone}$ for some
    $\typetwo$ and $\ctxthree$ by Lemma~\ref{lem:evalCtxtInversion},
    and substituting the term appearing in the right-hand side of
    $\qred$ we obtain that the desired typing $\qctwo:\typone$ is
    derivable.
  \end{varitemize}
This concludes the proof.
\end{proof}

A quantum closure $\qcone$ is in \emph{normal form} if for any distribution
$\distrone$ it does not hold that $\qcone\qred\distrone$. Reduction cannot
get stuck, as expected:
%\condmedskip
\begin{lemma}[Progress]
  If $[\qrone,\termone]:\typone$ is a quantum closure in normal form, then 
  $\termone$ is a value.
\end{lemma}
%\condmedskip
\begin{proof}
  This is just a case analysis on the form of $\termone$.
\end{proof}
The multistep relation $\qredmult$ is defined by the following set of rules:
$$
\AxiomC{}%$\qcone\in\normform$}
\UnaryInfC{$\qcone\qredmult\{\qcone^1\}$}
\DisplayProof
\quad
\AxiomC{$\qcone\qred\{\qctwo_1^{\probone_1},\ldots,\qctwo_n^{\probone_n}\}$}
\AxiomC{$\qctwo_i\qredmult\distrone_i$}
\BinaryInfC{$\qcone\qredmult\sum_{i=1}^n\probone_i\cdot\distrone_i$}
\DisplayProof
$$
\begin{prop}[Termination]
  If $\qcone:\typone$, there is a unique distribution $\distrone$ 
  such that $\qcone\qred\distrone$ and any quantum closure in the support of
  $\distrone$ is in normal form.
\end{prop}
\begin{proof}
  This is an induction on the size of $\termone$, where $\qcone=[\qrone,\termone]$.
\end{proof}
%%%%%%%%%%%%%%%%%%%%%%%%%%%%%%%%%%%%
\subsection{Translation into \QSMLL}
%%%%%%%%%%%%%%%%%%%%%%%%%%%%%%%%%%%%
\newcommand{\QtoSMLL}[1]{\langle #1\rangle}
\newcommand{\dtdone}{\Pi}
\newcommand{\dtdtwo}{\Psi}
\newcommand{\lkone}{l}
\newcommand{\pnone}{R}   % P lo usiamo gia' per le formule positive
\newcommand{\pntwo}{S}
\newcommand{\gpnone}{\mathbf{R}}
\newcommand{\gpntwo}{\mathbf{S}}

Type derivations of the $\lambda$-calculus which  we have just introduced will be mapped into  nets 
for an immediate generalization of \SMLLb, called \QSMLL. Specifically, we need to generalize \SMLLb\ in the 
following two ways:
\begin{varitemize}
\item
  \emph{Synchronization nodes are  labelled with a unitary operator} whose arity is the sum of the number of 
  atom occurrences in the involved formulas.
\item
  \emph{Boxes contain two nets.}
\end{varitemize} 
Types can be translated into formulas as follows:
$$
\QtoSMLL{\bool}=\QtoSMLL{\qbool}=1;\qquad
\QtoSMLL{\typone\lin\typtwo}=\QtoSMLL{\typone}^\bot\parr\QtoSMLL{\typtwo};
\qquad
\QtoSMLL{\typone\tens\typtwo}=\QtoSMLL{\typone}\tens\QtoSMLL{\typtwo}.
$$
%\emph{Quantum-Closed Type Derivations.} 
Any type derivation $\tdone$ with conclusions
$\tyg{\varone_1:\typone_1,\ldots,\varone_n:\typone_n,\qvarone_1,\ldots,\qvarone_m}{\termone}{\typtwo}$
can be translated into a net $\QtoSMLL{\tdone}$ of the following shape:
\begin{center}
\begin{tikzpicture}[ocenter]
% bang
\node at (0,0) [etic](pospal){};
% auxiliary ports
\node at (pospal.center)[below right=\hstalt and \stlar, etic](posauxlport){\tiny $\langle A_1\rangle^\bot$};
\node at (pospal.center)[below right=\hstalt and 3.75*\stlar, etic](posauxrport){\tiny $\langle A_n\rangle^\bot$};
\node at (pospal.center)[below right=\hstalt and 5*\stlar, etic](posauxpport){\tiny $\langle B\rangle$};
\node at (pospal.center)[below right=\hstalt and 2.375*\stlar, etic](label){\tiny $\cdots$};
% net g
\node at (pospal.center)[above right=0.4*\stalt and 2.275*\stlar, net](netg){$\langle\pi\rangle$};
\draw[nopol, in=90, out=-170](netg)to(posauxlport);
\draw[nopol, in=90, out=-10](netg)to(posauxrport);
\draw[nopol, in=90, out=0](netg)to(posauxpport);
\end{tikzpicture}
\end{center}
\begin{figure}
\fbox{
\begin{minipage}{.97\textwidth}
\begin{center}
$
\left\langle
\AxiomC{$\tdone\pof\tyg{\contone}{\termone}{\bool}$}
\AxiomC{$\tdtwo\pof\tyg{}{\termtwo}{\typone}$}
\AxiomC{$\tdthree\pof\tyg{}{\termthree}{\typone}$}
\TrinaryInfC{$\tyg{\contone}{\ifthen{\termone}{\termtwo}{\termthree}}{\typone}$}
\DisplayProof
\right\rangle
$
\qquad
$=$
\qquad
\begin{tikzpicture}[ocenter]
% bang
\node at (0,0) [etic](pospal){};
\node at (pospal.center)[right=2.135*\stlar, etic](posauxportu){};
\abox{pospal}{exbox}{\stboxlw}{\stboxrwaux+25pt}{\stboxh+6pt}
\node at (pospal.center) [etic](possym){$\bot$};
% auxiliary ports
\node at (pospal.center)[right=1.25*\stlar, etic](posauxportl){\tiny $\langle A\rangle$};
\node at (pospal.center)[right=3*\stlar, etic](posauxportr){\tiny $\langle A\rangle$};
\node at (pospal.center)[below right=0.5*\stalt and 2.135*\stlar, etic](posauxportd){\tiny $\langle A\rangle$};
% net g
\node at (pospal.center)[above right=0.75*\stalt and 1.25*\stlar, net](netl){$\langle\psi\rangle$};
\node at (pospal.center)[above right=0.75*\stalt and 3*\stlar, net](netr){$\langle\xi\rangle$};
\draw[nopol, in=90, out=-90](netl)to(posauxportl);
\draw[nopol, in=90, out=-90](netr)to(posauxportr);
\draw[nopol, in=90, out=-90](posauxportu)to(posauxportd);
% bot node
\node at (possym.center) [above=1.4*\altax, etic](botlink){\tiny $\mathsf{bot}$};
\draw[nopol, in=90, out=-90] (botlink) to (possym);
% principal port
%\node at (pospal.center)[below = \hstalt, etic](posprincport){\tiny $\bot$};
%\draw[nopol](possym)to(posprincport);
% rest of the net
\node at (pospal.center) [below left=1.3*\altax and \ilar, etic](cutlink){\tiny $\cut$};
%\node at (cutlink.center) [above left=1.3*\altax and \ilar, etic](onelink){\tiny $\mathsf{one}$};
\node at (cutlink.center)[above left=1.1*\stalt and \ilar, net](netguard){$\langle\pi\rangle$};
\draw[nopol, in=180, out=-90] (netguard) to (cutlink);
\draw[nopolrev, in=-90, out=0] (cutlink) to (possym);
% quantum register
%\node at (onelink)[above=.7*\stalt, reg](qreg){$Q$};
% links
%\draw[dashed, in=90, out=-90](qreg)to(onelink);
\node at (pospal.center)[below right=0.5*\stalt and 2.135*\stlar, etic](posauxportd){\tiny $\langle A\rangle$};
\end{tikzpicture}

\vspace{5pt}

$
\left\langle
\AxiomC{}
\UnaryInfC{$\tyg{\cdot}{\meas}{\qbool\lin\bool}$}
\DisplayProof
\right\rangle
$
\qquad
$=$
\qquad
\begin{tikzpicture}[ocenter]
% bang
\node at (0,0) [etic](pospal){};
\node at (pospal) [below=.1*\stalt, etic](pospaleps){};
\node at (pospal.center)[right=2.135*\stlar, etic](posauxportu){};
\abox{pospal}{exbox}{\stboxlw}{\stboxrwaux+20pt}{\stboxh+6pt}
\node at (pospal.center) [etic](possym){$\bot$};
% auxiliary ports
\node at (pospal.center)[right=1.25*\stlar, etic](posauxportl){\tiny $\langle \mathbb{B}\rangle$};
\node at (pospal.center)[right=3*\stlar, etic](posauxportr){\tiny $\langle \mathbb{B}\rangle$};
% net g
\node at (pospal.center)[above right=\stalt and 1.25*\stlar, etic](netlup){\tiny $\mathsf{one}$};
\node at (pospal.center)[above right=0.55*\stalt and 1.25*\stlar, etic](netldown){\tiny $\blacksquare$};
\node at (netldown)[right=0.25*\stalt, etic](netldownop){\tiny $\mathsf{I}$};
\node at (pospal.center)[above right=\stalt and 3*\stlar, etic](netrup){\tiny $\mathsf{one}$};
\node at (pospal.center)[above right=0.55*\stalt and 3*\stlar, etic](netrdown){\tiny $\blacksquare$};
\node at (netrdown)[right=0.25*\stalt, etic](netrdownop){\tiny $\mathsf{X}$};
\draw[nopol, in=90, out=-90](netlup)to(netldown);
\draw[nopol, in=90, out=-90](netldown)to(posauxportl);
\draw[nopol, in=90, out=-90](netrup)to(netrdown);
\draw[nopol, in=90, out=-90](netrdown)to(posauxportr);
% bot node
\node at (possym.center) [above=1.4*\altax, etic](botlink){\tiny $\mathsf{bot}$};
\draw[nopol, in=90, out=-90] (botlink) to (possym);
% par node
\node at (pospal.center)[below right=1.25*\stalt and 1.25*\stlar, etic] (parrPal){\tiny $\langle\mathbb{Q}\multimap\mathbb{B}\rangle$};
\lpar{parrPal}{pospaleps}{posauxportu}{par};
\end{tikzpicture}
\end{center}
\end{minipage}}
\caption{Translation of Type Derivation --- Some Interesting Cases}\label{fig:pntrans}
\end{figure}
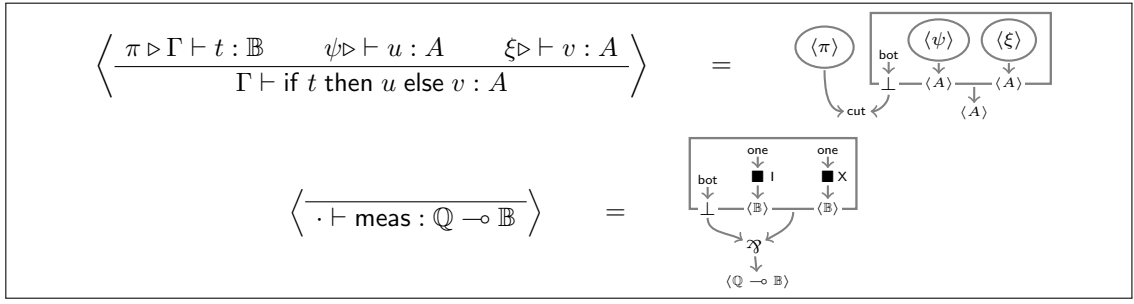
The construction is an induction on the structure of $\tdone$,
see Figure~\ref{fig:pntrans} for some interesting cases.

\noindent\emph{Extending the Translation to Quantum Closures.}
The translation scheme above can be extended to quantum closures, again by generalizing the notion of net; to a net 
we now associate a quantum register, as follows. A \emph{quantum net} is a pair $(\pnone,\qrone)$, where:
\begin{varitemize}
\item
  $\pnone$ is a \QSMLL\ net;
\item
  $\qrone$ is a normalized vector in $\hilb{\{\lkone_1,\ldots,\lkone_n\}}$, where $\lkone_1,\ldots,\lkone_n$ are 
    $\onelk$ links of $\pnone$ which are at depth $0$.
\end{varitemize}
%\emph{Type Derivations.} For the sake of the soundness proof, we can map any type derivation
%with conclusion $\tyg{\varone_1:\typone,\ldots,\varone_n:\typone,\qvarone_1,\ldots,\qvarone_m}{\termone}{\typtwo}$
%(together with a normalized vector $\qrone$ on $\hilb{\{\qvarone_1,\ldots,\qvarone_n\}}$
%can be translated into a quantum   net of the following shape:
The {generalization} above allows to map any quantum closure $[\qrone,\termone]$ into
the following quantum net 
\begin{center}
\begin{tikzpicture}[ocenter]
% bang
\node at (0,0) [etic](pospal){};
% auxiliary ports
\node at (pospal.center)[below right=\hstalt and \stlar, etic](posauxlport){\tiny $\langle A\rangle$};
%\node at (pospal.center)[below right=\hstalt and 3.75*\stlar, etic](posauxrport){\tiny $\langle A_n\rangle^\bot$};
%\node at (pospal.center)[below right=\hstalt and 5*\stlar, etic](posauxpport){\tiny $\langle B\rangle$};
%\node at (pospal.center)[below right=\hstalt and 2.375*\stlar, etic](label){\tiny $\cdots$};
% net g
\node at (pospal.center)[above right=0.4*\stalt and \stlar, net](netg){$\langle\pi\rangle$};
\draw[nopol, in=90, out=-90](netg)to(posauxlport);
%\draw[nopol, in=90, out=-10](netg)to(posauxrport);
%\draw[nopol, in=90, out=0](netg)to(posauxpport);
% ones
\node at (pospal.center) [etic, above left= .8*\stalt and .2*\stlar] (onespax){\tiny $\mathsf{one}$};
\node at (pospal.center) [etic, above left= 0.1*\stalt and .4*\stlar] (onesdots){\tiny $\vdots$};
\node at (pospal.center) [etic, below left= .2*\stalt and .2*\stlar] (onespal){\tiny $\mathsf{one}$};
\draw[nopol, in=135, out=0](onespax)to(netg);
\draw[nopol, in=-135, out=0](onespal)to(netg);
% quantum register
\node at (pospal.center)[above left=0.325*\stalt and 1.6*\stlar, reg](qreg){$Q$};
\draw[dashed, in=-180, out=90](qreg)to(onespax);
\draw[dashed, in=-180, out=-90](qreg)to(onespal);
\end{tikzpicture}
\end{center}
where $\tdone$ is a type derivation for $\termone$. Please observe the
way we make the correspondence between (some of the) $\mathsf{one}$
links and the quantum register $\qrone$ explicit by a dashed line.
This way, any quantum closure $\qcone$ can be associated to a quantum
net $\QtoSMLL{\qcone}$, and the map $\QtoSMLL{\cdot}$ can be
generalized to a function mapping distributions of quantum closures to
distributions of quantum nets.
%%%%%%%%%%%%%%%%%%%%%%%%%%%%%%%%%%%%%%%
\subsection{Normalization}
%%%%%%%%%%%%%%%%%%%%%%%%%%%%%%%%%%%%%%%
\newcommand{\qredpn}{\hookrightarrow}
\newcommand{\qredpnmult}{\leadsto}
\newcommand{\dpnone}{\mathscr{R}}
\newcommand{\dpntwo}{\mathscr{Q}}
\newcommand{\func}[1]{[#1]}
\newcommand{\cutpn}[3]{\mathsf{cut}(#1,#2.#3)}
The additional features of \QSMLL\ require an adaptation of the cut elimination procedure.
%Cut-elimination only makes sense if defined on \emph{all} quantum   nets. 
Moreover, reduction becomes probabilistic. We are then forced to consider distributions
over  nets, which we note $\dpnone,\dpntwo,\ldots$, in the same way as the distributions
over quantum closures we considered in Section~\ref{sect:qcos}. The new reduction rules are the
ones in Figure~\ref{fig:redqpn}.
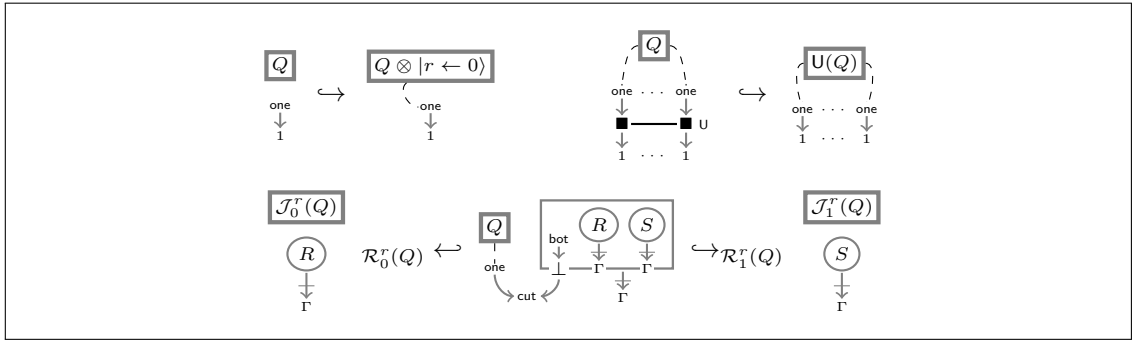
\begin{figure}
\begin{center}
\fbox{
\begin{minipage}{\figurewidth}
\begin{center}
\vspace{7pt}
\begin{center}
\begin{tikzpicture}[ocenter]
\node at (0,0) [etic](onePal){\tiny $1$};
\lone{onePal}{one};
\node at (onePal.center)[above=1.2*\stalt, reg](qreg){$Q$};
\end{tikzpicture}
\hspace{1pt}
$\hookrightarrow$
\hspace{1pt}
\begin{tikzpicture}[ocenter]
\node at (0,0) [etic](onePal){\tiny $1$};
\lone{onePal}{one};
\node at (onePal.center)[above=1.2*\stalt, reg](qreg){$Q\tens|r\leftarrow 0\rangle$};
\draw[dashed, in=180, out=215](qreg)to(one);
\end{tikzpicture}
\qquad\qquad
\begin{tikzpicture}[ocenter]
\node at (0,0) [etic](syplftPal){\tiny $1$};
\node at (syplftPal.center) [etic, above= \stalt] (syplftPax){\tiny $\mathsf{one}$};
\lsync{syplftPal}{syplftPax}{syplft};
\node at (syplftPal) [right= 3*\altax, etic](syprhtPal){\tiny $1$};
\node at (syprhtPal.center) [etic, above= \stalt] (syprhtPax){\tiny $\mathsf{one}$};
\lsync{syprhtPal}{syprhtPax}{syprht};
% dots
\node at (syplftPal.center) [right=1.5*\altax, etic](uldots){\tiny $\cdots$};
\node at (syplftPax.center) [right=1.5*\altax, etic](lldots){\tiny $\cdots$};
% links
\draw[thick] (syplft) to (syprht);
% quantum register
\node at (syplftPax.center)[above right=.8*\stalt and 1.5*\altax, reg](qreg){$Q$};
% links
\node at (syplftPax) [above=.4*\altax](syplftPaxup){};
\node at (syprhtPax) [above=.4*\altax](syprhtPaxup){};
\draw[dashed, in=-90, out=-180](qreg)to(syplftPaxup);
\draw[dashed, in=-90, out=0](qreg)to(syprhtPaxup);
% unitary operator
\node at (syprht) [right=.15*\altax](unop){\tiny $\mathsf{U}$};
\end{tikzpicture}
\hspace{1pt}
$\hookrightarrow$
\hspace{1pt}
\begin{tikzpicture}[ocenter]
\node at (0,0) [etic](onelPal){\tiny $1$};
\lone{onelPal}{onel};
\node at (0,0) [right=3*\altax, etic](onerPal){\tiny $1$};
\lone{onerPal}{oner};
% dots
\node at (onel.center) [right=1.5*\altax, etic](onedots){\tiny $\cdots$};
\node at (onelPal.center) [right=1.5*\altax, etic](onesdots){\tiny $\cdots$};
% quantum register
\node at (onel.center)[above right=.8*\stalt and 1.5*\altax, reg](qreg){$\mathsf{U}(Q)$};
% links
\node at (onel) [above=.4*\altax](onelup){};
\node at (oner) [above=.4*\altax](onerup){};
\draw[dashed, in=-90, out=-180](qreg)to(onelup);
\draw[dashed, in=-90, out=0](qreg)to(onerup);
\end{tikzpicture}

\vspace{10pt}

\begin{tikzpicture}[ocenter]
%bang
\node at (0,0) [ etic](pospal){};
%auxiliary ports
\node at (pospal.center)[below right=\hstalt and \stlar, etic](posauxkport){\tiny $\Gamma$};
%net g
\node at (pospal.center)[above right=0.4*\stalt and \stlar, net](netg){$R$};
\draw[nopolgen, in=90, out=-90](netg)to(posauxkport);
% quantum register
\node at (netg.center)[above=.8*\stalt, reg](qreg){$\proj{Q}{r}{0}$};
\end{tikzpicture}
$\;{}_{\pofobs{Q}{r}{0}}\hookleftarrow\;$
\begin{tikzpicture}[ocenter]
% bang
\node at (0,0) [etic](pospal){};
\node at (pospal.center)[right=1.6*\stlar, etic](posauxportu){};
\abox{pospal}{exbox}{\stboxlw}{\stboxrwaux+7pt}{\stboxh+5pt}
\node at (pospal.center) [etic](possym){$\bot$};
% auxiliary ports
\node at (pospal.center)[right=\stlar, etic](posauxportl){\tiny $\Gamma$};
\node at (pospal.center)[right=2.2*\stlar, etic](posauxportr){\tiny $\Gamma$};
\node at (pospal.center)[below right=0.5*\stalt and 1.6*\stlar, etic](posauxportd){\tiny $\Gamma$};
% net g
\node at (pospal.center)[above right=0.75*\stalt and \stlar, net](netl){$R$};
\node at (pospal.center)[above right=0.75*\stalt and 2.2*\stlar, net](netr){$S$};
\draw[nopolgen, in=90, out=-90](netl)to(posauxportl);
\draw[nopolgen, in=90, out=-90](netr)to(posauxportr);
\draw[nopolgen, in=90, out=-90](posauxportu)to(posauxportd);
% bot node
\node at (possym.center) [above=1.4*\altax, etic](botlink){\tiny $\mathsf{bot}$};
\draw[nopol, in=90, out=-90] (botlink) to (possym);
% principal port
%\node at (pospal.center)[below = \hstalt, etic](posprincport){\tiny $\bot$};
%\draw[nopol](possym)to(posprincport);
% rest of the net
\node at (pospal.center) [below left=1.3*\altax and \ilar, etic](cutlink){\tiny $\cut$};
\node at (cutlink.center) [above left=1.3*\altax and \ilar, etic](onelink){\tiny $\mathsf{one}$};
\draw[nopol, in=180, out=-90] (onelink) to (cutlink);
\draw[nopolrev, in=-90, out=0] (cutlink) to (possym);
% quantum register
\node at (onelink)[above=.7*\stalt, reg](qreg){$Q$};
% links
\draw[dashed, in=90, out=-90](qreg)to(onelink);
\end{tikzpicture}
$\;\hookrightarrow_{\pofobs{Q}{r}{1}}\;$
\begin{tikzpicture}[ocenter]
%bang
\node at (0,0) [ etic](pospal){};
%auxiliary ports
\node at (pospal.center)[below right=\hstalt and \stlar, etic](posauxkport){\tiny $\Gamma$};
%net g
\node at (pospal.center)[above right=0.4*\stalt and \stlar, net](netg){$S$};
\draw[nopolgen, in=90, out=-90](netg)to(posauxkport);
% quantum register
\node at (netg.center)[above=.8*\stalt, reg](qreg){$\proj{Q}{r}{1}$};
\end{tikzpicture}
\end{center}
\vspace{7pt}
\end{center}
\end{minipage}}
\condnocaptionrule\caption{Reduction --- New Rules}\label{fig:redqpn}
\end{center}
\end{figure}
Analogously to what happens for $\lambda$-terms, then, two reduction relations
can be defined on nets, namely a one-step reduction relation $\qredpn$ and a multi-step
reduction relation $\qredpnmult$, both between  nets and distributions over  nets.
The main result of this section is the following:
%\condmedskip
\begin{prop}[Simulation]\label{prop:simtermspn}
  If $\qcone\qredmult\distrone$, then 
  $\QtoSMLL{\qcone}\qredpnmult\QtoSMLL{\distrone}$.
\end{prop}
%\condmedskip
%%%%%%%%%%%%%%%%%%%%%%%%%%
\subsubsection*{The Proof}
%%%%%%%%%%%%%%%%%%%%%%%%%%
\newcommand{\pnthree}{S}
\newcommand{\pnfour}{T}
\newcommand{\dpnthree}{\mathscr{S}} %% R and Q are misleading
\newcommand{\probthree}{r}
First of all, one can prove the following substitution lemma, where $\cutpn{\tdone}{\varone}{\tdtwo}$ 
is the  net obtained by $\QtoSMLL{\tdone}$ and $\QtoSMLL{\tdtwo}$, cutting the link for $\varone$
in the former to the conclusion for the latter.
\begin{lemma}[Subsitution Lemma]\label{lemma:substitution}
  Suppose $\tdone\pof\tyg{\varone:\typone,\qvarone_1,\ldots,\qvarone_n}{\termone}{\typtwo}$
  and $\tdtwo\pof\tyg{\qvartwo_1,\ldots,\qvartwo}{\termtwo}{\typone}$. Then
  $\cutpn{\tdone}{\varone}{\tdtwo}
  \qredpnmult
  \{\QtoSMLL{\tdone\{\tdtwo/\varone\}}^1\}$.
\end{lemma}
\begin{proof}
  By induction on the structure of $\tdone$.
\end{proof}
We then need another auxiliary result:
\begin{lemma}\label{lem:multRule}
  If $\pnone \qredpnmult \{\pntwo_1^{\probone_1}, \dots, \pntwo_n^{\probone_n}\}$
  and $\pntwo_i \qredpnmult \dpnone_i$,
  then $\pnone \qredpnmult \sum_{i=1}^{n} p_i \dpnone_i$.
\end{lemma}
\begin{proof}
  By induction on the derivation of
  $\pnone \qredpnmult \{\pntwo_1^{\probone_1}, \dots, \pntwo_n^{\probone_n}\}$.
  \begin{varitemize}
  \item If the derivation is
    $$
    \AxiomC{}
    \UnaryInfC{$\pnone \qredpnmult \{\pnone^{1}\}$}
    \DisplayProof
    $$
    and $\pnone \qredpnmult \dpnone$,
    then clearly $\pnone \qredpnmult 1 \cdot \dpnone = \dpnone$.
  \item If the derivation is
    $$
    \AxiomC{$\pnone \qredpn \{\pnthree_1^{\probtwo_1}, \dots, \pnthree_m^{\probtwo_m}\}$}
        \AxiomC{$\pnthree_i \qredpnmult \dpnthree_i$}
    \BinaryInfC{$\pnone \qredpnmult \sum_{i=1}^{m} \probtwo_i \dpnthree_i
      = \{\pntwo_1^{\probone_1}, \dots, \pntwo_n^{\probone_n}\}$}
    \DisplayProof
    $$
    then there are two possibilities.
    Note that $\pntwo_i \qredpnmult \dpnone_i$ holds by assumption.
    \begin{varitemize}
    \item If
      $$
      \AxiomC{$\pnone \qredpn \{\pnthree^{1}\}$}
      \AxiomC{$\pnthree \qredpnmult \dpnthree
        = \{\pntwo_1^{\probone_1}, \dots, \pntwo_n^{\probone_n}\}$}
      \BinaryInfC{$\pnone \qredpnmult 1 \cdot \dpnthree
        = \{\pntwo_1^{\probone_1}, \dots, \pntwo_n^{\probone_n}\}$}
      \DisplayProof
      $$
      then by induction hypothesis $\pnthree \qredpnmult \sum_{i=1}^{n} \probone_i \dpnone_i$.
      Hence
      $$
      \AxiomC{$\pnone \qredpn \{\pnthree^{1}\}$}
      \AxiomC{$\pnthree \qredpnmult \sum_{i=1}^{n} \probone_i \dpnone_i$}
      \BinaryInfC{$\pnone \qredpnmult \sum_{i=1}^{n} \probone_i \dpnone_i$}
      \DisplayProof
      $$

    \item If
      $$
      \small
      \AxiomC{$\pnone \qredpn \{\pnthree_1^{\probtwo}, \pnthree_2^{1 - \probtwo}\}$}
          \AxiomC{$\pnthree_1 \qredpnmult \dpnthree_1
            = \{\pntwo_1^{\probthree_1}, \dots, \pntwo_k^{\probthree_k}\}$}
          \noLine
          \UnaryInfC{$\pnthree_2 \qredpnmult \dpnthree_2
            = \{\pntwo_{k+1}^{\probthree_{k+1}}, \dots, \pntwo_n^{\probthree_n}\}$}
      \BinaryInfC{$\pnone \qredpnmult
        \probtwo \cdot \dpnthree_1 + (1 - \probtwo) \cdot \dpnthree_2
        = \{\pntwo_1^{\probone_1}, \dots, \pntwo_n^{\probone_n}\}$}
      \DisplayProof
      $$
      where $1 \le k < n$, then $\probone_i$ satisfies 
      $$
      \small
      \probone_i =
      \begin{cases}
        \probtwo \probthree_i \enspace \text{if $i \le k$;}\\
        (1 - \probtwo) \probthree_i \enspace \text{if $i > k$}.
      \end{cases}
      $$
      By induction hypothesis $\pnthree_1 \qredpnmult \sum_{i=1}^{k} \probthree_i \dpnone_i$
      and $\pnthree_2 \qredpnmult \sum_{i=k+1}^{n} \probthree_i \dpnone_i$, and
      $$
      \AxiomC{$\pnone \qredpn \{\pnthree_1^{\probtwo}, \pnthree_2^{1 - \probtwo}\}$}
        \AxiomC{$\pnthree_1 \qredpnmult \sum_{i=1}^{k} \probthree_i \dpnone_i$}
        \noLine
        \UnaryInfC{$\pnthree_2 \qredpnmult \sum_{i=k+1}^{n} \probthree_i \dpnone_i$}
      \BinaryInfC{$\pnone \qredpnmult
        \probtwo \cdot \sum_{i=1}^{k} \probthree_i \dpnone_i
        + (1 - \probtwo) \cdot \sum_{i=k+1}^{n} \probthree_i \dpnone_i$}
      \DisplayProof
      $$      
      This precisely says $\pnone \qredpnmult \sum_{i=1}^{n} \probone_i \dpnone_i$:
      {\small
        \begin{align*}
          &\probtwo \cdot \sum_{i=1}^{k} \probthree_i \dpnone_i
          + (1 - \probtwo) \cdot \sum_{i=k+1}^{n} \probthree_i \dpnone_i\\
          =& \sum_{i=1}^{k} \probtwo \probthree_i \dpnone_i
          + \sum_{i=k+1}^{n} (1 - \probtwo) \probthree_i \dpnone_i\\
          =& \sum_{i=1}^{k} \probone_i \dpnone_i
          + \sum_{i=k+1}^{n} \probone_i \dpnone_i
          = \sum_{i=1}^{n} \probone_i \dpnone_i \enspace.
        \end{align*}
      }%
    \end{varitemize}
  \end{varitemize}
This concludes the proof.
\end{proof}
The next lemma uses Lemma~\ref{lemma:substitution} in an essential way:
\begin{lemma}\label{lem:oneStepSimulation}
  If $C \qred \{D_1^{p_1}, \dots, D_n^{p_n}\}$
  then $\QtoSMLL{C} \qredpnmult \{\QtoSMLL{D_1}^{p_1}, \dots, \QtoSMLL{D_n}^{p_n}\}$.
  \hfill $\Box$
\end{lemma}
\begin{proof}
  By case analysis.
  It is straightforward; in the cases of
  $[\qrone,\apecd{(\abstr{\varone}{\termone})\termtwo}]\qred\{[\qrone,\apecd{\subst{\termone}{\varone}{\termtwo}}]^1\}$ and
$[\qrone,\apecd{\llet{\varone}{\vartwo}{\pair{\termone}{\termtwo}}{\termthree}}]
  \qred\{[\qrone,\apecd{\subst{\termthree}{\varone,\vartwo}{\termone,\termtwo}}]^1\}$,
  Lemma~\ref{lemma:substitution} guarantees the statement: the conclusion of $\tdone\{\tdtwo/\varone$
  is necessarily $\termone\{\termtwo/\varone\}$ and
  the derivation of $\termone\{\termtwo/\varone\}$ is unique.
\end{proof}
We are now ready to prove the Simulation result. This proceeds
by a simple case analysis on the derivation of $\qcone\qredmult\distrone$.
\begin{varitemize}
\item 
  The case of $\qcone\qredmult \{\qcone^1\}$ is clear.
\item 
  In the other case,
  $$
  \AxiomC{$\qcone\qred\{\qctwo_1^{\probone_1},\ldots,\qctwo_n^{\probone_n}\}$}
  \AxiomC{$\qctwo_i\qredmult\distrone_i$}
  \BinaryInfC{$\qcone\qredmult\sum_{i=1}^n\probone_i\cdot\distrone_i$}
  \DisplayProof
  $$
  By Lemma~\ref{lem:oneStepSimulation}
  $\QtoSMLL{C} \qredpnmult \{\QtoSMLL{D_1}^{p_1}, \dots, \QtoSMLL{D_n}^{p_n}\}$.
  By induction hypothesis $\QtoSMLL{D_i} \qredpnmult \QtoSMLL{\mathscr{D}_i}$.
  Hence by Lemma~\ref{lem:multRule} 
    $\QtoSMLL{C} \qredpnmult \sum_i p_i \QtoSMLL{\mathscr{D}_i}
  = \QtoSMLL{\sum_i\probone_i\cdot\distrone_i}$.
\end{varitemize}
%%%%%%%%%%%%%%%%%%%%%%%%%%%%%%%%%%%%%%
\section{Computing  with the \QSIAM} \label{qsiam}
%%%%%%%%%%%%%%%%%%%%%%%%%%%%%%%%%%%%%%

The \SIAM\ can be generalized to the \QSIAM, an abstract machine for
quantum nets. Most definitions about the \QSIAM\ 
are inherited from those about the \SIAM. In particular, positions
and sets of positions such as $\pos{\pnone}$, $\botcon{\pnone}$, etc. are 
defined exactly in the same way.
%Before defining it, we need to give some
%auxiliary definitions, where we assume to work with a  net $\pnone$:
%\begin{varitemize}
%\item
%  The set of positions in $\pnone$, namely the set of all occurrences of the atoms $1$ and $\bot$ in
%  $\pnone$, is indicated with $\pos{\pnone}$;
%\item
%  $\botcon{\pnone}$ is the subset of $\pos{\pnone}$ of all those positions corresponding to
%  occurrences of $\bot$ in the conclusion of $\pnone$. Similarly for occurrences
%  of $1$ in the conclusion of $\pnone$, and in this case the set is $\onecon{\pnone}$ 
%\item
%  $\oneany{\pnone}$ is another subset of $\pos{\pnone}$ and includes all the \emph{constant positions}
%  in $\pnone$, namely those occurrences of the formula $1$ next to a $1$ node.
%\item
%  Finally, $\botbox{\pnone}$ is again a subset of $\pos{\pnone}$, this time including those
%  occurrences of $\bot$ which act as ``guards'' of boxes.
%\end{varitemize}

\noindent\emph{Statics.}
The states of $M_{\gpnone}$, the interactive machine interpreting the
quantum net $\gpnone=(\pnone,\qrone)$, are the pairs
$\stone=(\tksone,\qrtwo)$ where
\begin{varitemize}
\item 
  $\tksone$ is a function from $\pssetone \cup \botcon{\pnone}$ to
  $\pos{\pnone}$, where $\pssetone$ is a subset of
  $\oneany{\pnone}$. As in the \SIAM, the role of $\tksone$ is to
  capture \emph{where} the $|\pssetone|$ tokens started from (namely
  the positions in $\pssetone$) and where they are (namely the
  positions in $\tksone(\pssetone)$).
\item
  $\qrtwo$ is a quantum register in $\hilb{\sbsts{\pssetone}}$.
\end{varitemize}
The set of states of $M_{\gpnone}$ is denoted by
$\states{\gpnone}$. The notions of an initial state, a final state,
and an active state are very close to the ones given for the \SIAM.
The set of all final states of $M_{\gpnone}$ is denoted by
$\fstates{\gpnone}$.
% We are interested in isolating a class of states enjoying nice properties, that will later be
% proved to be invariants. More specifically, not any position should appear in the domain and
% codomain of $\tksone$ when $(\tksone,\qrone)$ is a state. Suppose that 
% $\posone\in\pos{\pnone}$, and suppose that $\postwo_1,\ldots,\postwo_n$ are
% the inner positions for $\bot$ in boxes encosing $\posone$, i.e.,
% \begin{center}
% \TODO{Proof-net explaining the situation.}
% \end{center}
% Then $\posone$ is said to be \emph{active} in $(\tksone,\qrone)$  if there
% is $\{\posthree_1,\ldots,\posthree_n\}\subseteq\pssetone$ such
% that $\tksone(\posthree_i)=\postwo_i$ and 
% and $\qrone=\qrtwo\tens
% \ket{\posthree_1\leftarrow\bitone_1}\tens\cdots\tens\ket{\posthree_n\leftarrow\bitone_n}$.
% One would like all positions in the domain and codomain of $\tksone$ to be \emph{active},
% as this corresponds to the ensuring that all tokens make a journey in the appropriate slice.
% A state $(\tksone,\qrone)\in\states{\gpnone}$ is said to be \emph{initial}  if $\tksone$
% is the identity on $\botcon{\pnone}$, and \emph{final}  if the range of $\tksone$ is
% exactly a subset of $\botbox{\pnone}\cup\onecon{\pnone}$. The sets of initial and final
% states are denoted as $\istates{\pnone}$ and $\fstates{\pnone}$, respectively.

\noindent\emph{Dynamics.} It is now time to describe how states
evolve. This takes again the form of a relation $\trrel$ between
$\states{\gpnone}$ and \emph{finite distributions} over the same
set. Rules follow quite closely the rules for the \SIAM.
%Each transition rule shows how a state can evolve  by showing how the codomain
%of $\tksone$ changes, how the domain of $\tksone$ is  extended, and how the
%quantum register $\qrone$ is altered. 
%The relation $\trrel$ captures state evolution in \emph{one} step. 
Extending $\trrel$ into a multi-step relation $\trrelmult$ can be done
in the usual, standard, way:
$$
{\footnotesize
\AxiomC{}
\UnaryInfC{$\stone\trrelmult\emptyset$}
\DisplayProof
\quad
\AxiomC{$\stone\in\fstates{\gpnone}$}
\UnaryInfC{$\stone\trrelmult\{\stone^1\}$}
\DisplayProof
\quad
\AxiomC{$\stone\trrel\{\sttwo_1^{\probone_1},\ldots,\sttwo_n^{\probone_n}\}$}
\AxiomC{$\sttwo_i\trrelmult\distrone_i$}
%% \BinaryInfC{$\stone\trrelmult\sum_{i=1}^n\frac{1}{\probone_i}\distrone_i$}
\BinaryInfC{$\stone\trrelmult\sum_{i=1}^n\probone_i\cdot\distrone_i$}
\DisplayProof}
$$
The semantics of any state $\stone$ is simply
$$
\sem{\stone}=\sup_{\stone\trrelmult\distrone}\distrone.
$$
Analogously to what has been done for \SMLLb, we can define the object
$\func{\gpnone}$ computed by the \QSIAM\ for $\gpnone$. 
To do that, let us observe the following:
\begin{varitemize}
\item
  The initial state is not anymore unique: each initial state corresponds to an
  element of $\hilb{\botcon{\gpnone}}$.
\item
  To each initial state $\stone$ we can  associate %put in correspondence
  a distribution of final states, each consisting of a partial 
  injection from $\botcon{\gpnone}$ to $\onecon{\gpnone}$, and
  an element of $\hilb{\onecon{\gpnone}}$. If the execution get stuck,
  the result is taken to be $\Box$.
\end{varitemize}
This correspondence is what is taken as the interpretation
$\func{\gpnone}$ of $\gpnone$. In other words: 
$$
\func{\gpnone}:\hilb{\botcon{\gpnone}}\rightarrow
    \quad\distrs{(\botcon{\gpnone}\rightharpoonup\onecon{\gpnone})\times\hilb{\onecon{\gpnone}}\cup\{\Box\}}.
$$
We can then prove the following:
\begin{theorem}[Soundness]\label{theo:soundnessQSIAM}
  If $\gpnone\qredpnmult\dpnone$, then $\func{\gpnone}=\func{\dpnone}$.
\end{theorem}
The way we prove Theorem~\ref{theo:soundnessQSIAM} is by way of a
careful analysis of paths.  The crucial step consists in proving that
whenever
$\gpnone\qredpn\{\gpntwo_1^{\probone_1},\ldots,\gpntwo_n^{\probone_n}\}$,
then there is a function
$$
\Phi:\states{\gpnone}\rightarrow(\states{\gpntwo_1}\times\RR)\times\ldots\times(\states{\gpntwo_n}\times\RR)
$$ 
satisfying the following properties:
\begin{varitemize}
\item
  On the one hand, if $\stone\in\states{\gpnone}$ and
  $\stone\trrelmult\distrone$, then
  $\Phi(\stone)=((\sttwo_1,\probtwo_1),\ldots,(\sttwo_n,\probtwo_n))$
  and for every $1\leq i\leq n$, $\sttwo_i\trrelmult\distrtwo_i$ where
  $\distrone\leq\sum_{i=1}^n\probtwo_i\distrtwo_i$
\item
  On the other, if
  $\Phi(\stone)=((\sttwo_1,\probtwo_1),\ldots,(\sttwo_n,\probtwo_n))$,
  and for every $1\leq i\leq n$, $\sttwo_i\trrelmult\distrtwo_i$ then
  $\stone\trrelmult\distrone$ where
  $\distrone\geq\sum_{i=1}^n\probtwo_i\distrtwo_i$
\item
  $\Phi$ maps an initial state of $\gpnone$ to the corresponding
  initial state of $\gpntwo_1,\ldots,\gpntwo_n$, each with probability
  $\probone_1,\ldots,\probone_n$. Similarly for final states.
\end{varitemize}
For each reduction rule, one can define the function $\Phi$ and prove the three properties
above, the first two by induction on the structure of the underlying derivation.

Combining Theorem~\ref{theo:soundnessQSIAM} and Proposition~\ref{prop:simtermspn}, one
gets that the \QSIAM\ is a model of computation which is adequate with respect to the
operational semantics introduced in Section~\ref{sect:qcos}:
\begin{corollary}
  If $\qcone\qredmult\distrone$, then $\func{\QtoSMLL{\qcone}}=\func{\QtoSMLL{\distrone}}$.
\end{corollary}

%%%%%%%%%%%%%%%%%%%%%%%
\section{Discussion}
%%%%%%%%%%%%%%%%%%%%%%%
\paragraph{Sequentialization.}
The graphical calculus we propose here does not have a
sequent calculus counterpart, at least not a standard one. One would
need to add some extra information (for example, a coherence relation
on atoms), to express the fact that atoms from different axioms can be
synchronized (or entangled).

A variation on \SMLL\ which would admit sequentialization into a sequent calculus
is obtained by forcing all sync links to be \emph{unary}, i.e., to
have a single premiss and a single conclusion.  To such a link one can
easily associate a sequent calculus rule, or a term
derivation. However, cut elimination only holds if the net is
closed. Such a solution was explored in \cite{tristan}. A more
sophisticated but somehow similar solution has been proposed by one of
the anonymous referees, who suggested sync links to be identified with a
new kind of \emph{synchronous} cut link.

There is actually a trade-off between two desirable results here:
sequentialization and cut elimination.  A good example is the one in
Figure \ref{special_main}, which can be understood as the net
associated to a term in the form
$\llet{\varthree}{\varfour}{(\uopone\pair{\varone}{\vartwo})}{\pair{\varthree}{\varfour}}$.
The net on the l.h.s.  cannot be reduced further if we limit ourselves
to unary sink links. On the other hand, to this net we can associate a
sequent calculus proof, while it is not the case for the net on the
r.h.s.  In this paper, we prefer to have cut elimination without
conditions, because cut elimination gives us a tool to deal with
deadlock freedom.
%One can also observe that if we only work with \SMLL\ and closed nets,
%the normal form has no sync links, hence we can always associate a
%sequent calculus proof (a term derivation).

\paragraph{Compiling Terms into Circuits.}
The \QSIAM\ machine is definitely a quantum automaton: unitary
transformations and measurements are performed while visiting the
net. It would also be interesting, especially in the measurement-free
case, to design token machines which \emph{extract} a quantum circuit
from a net instead of executing it on-the-fly. The obtained machine
would of course be sound only in the absence of the sync elimination
rule, so that in the normal form (which would essentially be a
\SMLLcf\ net) the unitary gates remain explicit. By the way, having
this option plays in favor of the choice discussed in the paragraph
above, since it is indeed \emph{cut elimination} which allows us to
prove the absence of deadlocks.

%%%%%%%%%%%%%%%%%%%%%%%%
\section{Conclusions}
%%%%%%%%%%%%%%%%%%%%%%%%
This work can be seen as the first step towards making Interaction
Abstract Machines a more general model of computation in which
not only parallelism, but also synchronization, can take place. Interestingly,
this is done with tools coming from proof-theory, namely proof-nets.
Noticeably, desirable properties like termination and deadlock
freedom are byproduct of correctness.

The main weakness of this work is that the underlying logical system,
namely \MLL, is of limited expressive power. Adding exponential
connectives to \SMLL\ is quite natural, and has not been done here
only for the sake of simplicity. Another point worth investigating is
certainly a further analysis on the nature of synchronization,
and in particular on the possibility of synchronizing over formulas
neither strictly positive nor strictly negative. In general, this can
lead to deadlocks, but how about isolating a class of \emph{safe}
formulas?

\section*{Acknowledgments}
The first author is supported by the project ANR-12IS02001
``PACE''. The second author is supported by the project
ANR-2010-BLANC-021301 ``LOGOI''. The third and fourth authors are
supported by Grants-in-Aid for Young Scientists (A) No.\ 24680001,
JSPS, and by Aihara Innovative Mathematical Modeling Project, FIRST
Program, JSPS/CSTP.
 
\bibliographystyle{abbrv}
\bibliography{biblio}
\end{document}